\documentclass[11pt,a4paper]{article}
 \usepackage{amsmath, amsfonts, amssymb}
\usepackage{mathtools}
\usepackage{a4wide}
\usepackage{parskip}
\usepackage{enumitem}
\usepackage{xcolor}

\usepackage{hyperref}
\usepackage{booktabs}% nicer horizontal lines
\usepackage{caption}% fix vertical spacing of table captions
\usepackage{siunitx}% align numbers at the decimal point
\usepackage{tabularx}
\usepackage{graphicx}
\usepackage{adjustbox}
\usepackage{multirow}
\usepackage{amsthm}
\usepackage{bm}
\usepackage{lscape}
\usepackage{comment}
\usepackage{float}
\usepackage[noadjust]{cite}

\newtheorem{theorem}{Theorem}[section]
\newtheorem{lemma}[theorem]{Lemma}
\newtheorem{corollary}[theorem]{Corollary}
\theoremstyle{definition}
\newtheorem{definition}[theorem]{Definition}
\newtheorem{example}{Example}
\newtheorem{proposition}[theorem]{Proposition}
\theoremstyle{remark}
\newtheorem{remark}[theorem]{Remark}

\numberwithin{equation}{section}

\newcommand{\Aut}{\textnormal{Aut}}

\newcommand{\GL}{\textnormal{GL}}
\newcommand{\Id}{\textnormal{Id}}
\usepackage{tikz,xcolor,hyperref}

\definecolor{lime}{HTML}{A6CE39}
\DeclareRobustCommand{\orcidicon}{%
	\begin{tikzpicture}
		\draw[lime, fill=lime] (0,0) 
		circle [radius=0.16] 
		node[white] {{\fontfamily{qag}\selectfont \tiny ID}};
		\draw[white, fill=white] (-0.0625,0.095) 
		circle [radius=0.007];
	\end{tikzpicture}
	\hspace{-2mm}
}

\foreach \x in {A, ..., Z}{%
	\expandafter\xdef\csname orcid\x\endcsname{\noexpand\href{https://orcid.org/\csname orcidauthor\x\endcsname}{\noexpand\orcidicon}}
}

\begin{document}
	\date{}
	{\vspace{0.01in}
		\title{$(\Theta, \Delta_\Theta, \bm{a})$-cyclic codes over $\mathbb{F}_q^l$ and their applications in the construction of quantum codes} 
		\author{{\bf Akanksha\footnote{email: {\tt akankshafzd8@gmail.com}}\orcidA{},\; Anuj Kumar Bhagat\footnote{email: {\tt anujkumarbhagat632@gmail.com}}\orcidB{},\; and \bf Ritumoni Sarma\footnote{	email: {\tt ritumoni407@gmail.com}}\orcidC{}} \\ Department of Mathematics,\\ Indian Institute of Technology Delhi,\\Hauz Khas, New Delhi-110016, India. }
		\maketitle
		\begin{abstract}
  \noindent
  In this article, for a finite field $\mathbb{F}_q$ and a natural number $l,$ let $\mathcal{R}$ denote the product ring $\mathbb{F}_q^l.$ Firstly, for an automorphism $\Theta$ of $\mathcal{R},$ a $\Theta$-derivation $\Delta_\Theta$ of $\mathcal{R}$ and for a unit $\bm{a}$ in $\mathcal{R},$ we study $(\Theta, \Delta_\Theta, \bm{a})$-cyclic codes over $\mathcal{R}.$ In this direction, we give an algebraic characterization of a $(\Theta, \Delta_\Theta, \bm{a})$-cyclic code over $\mathcal{R}$, determine its generator polynomial, and find its decomposition over $\mathbb{F}_q.$ Secondly, we give a necessary and sufficient condition for a $(\Theta, 0, \bm{a})$-cyclic code to be Euclidean dual-containing code over $\mathcal{R}.$ Thirdly, we study Gray maps and obtain several MDS and optimal linear codes over $\mathbb{F}_q$ as Gray images of $(\Theta, \Delta_\Theta, \bm{a})$-cyclic codes over $\mathcal{R}.$ Moreover, we determine orthogonality preserving Gray maps and construct Euclidean dual-containing codes with good parameters. Lastly, as an application, we construct MDS and almost MDS quantum codes by employing the Euclidean dual-containing and annihilator dual-containing CSS constructions.    

			\medskip
			
			
			\medskip
			
			
		\end{abstract}
\section{Introduction}\label{Section 1}
Linear codes play an essential role in the correction and detection of errors during the transmission of information. Cyclic codes, a class of linear codes, have gained attention due to their nice polynomial representation and easy encoding-decoding. While these codes were initially studied over finite fields, Hammons et al. \cite{Hammons} in 1994 highlighted the potential for constructing numerous binary codes using cyclic codes and the Gray map over the ring $\mathbb{Z}_4$. Since then, there has been extensive interest in coding theory over finite rings, of which much research has focused on the codes over finite commutative rings, for instance, see \cite{cao2020class} and \cite{cao2020construction}.
\par
As the number of cyclic codes and the generator polynomial of a cyclic code depend on the factorization of a polynomial, it is logical to extend the study of cyclic codes over a structure that allows polynomials to have more than one factorization. In this context, one of the best choices is the skew polynomial ring because it allows an easy extension of cyclic codes to a non-commutative setup as skew cyclic codes. The work of Boucher et al. \cite{boucher2007skew} prompted coding theorists to investigate codes over skew polynomial rings. Following this, various studies were conducted on skew cyclic codes and skew constacyclic codes over many finite chain rings and finite non-chain rings. For instance, research on skew cyclic codes was performed over the ring $\mathbb{F}_2 + v\mathbb{F}_2$, where $v^2 = v$ in \cite{abualrub2012theta}, as well as over $\mathbb{F}_p + u\mathbb{F}_p$, with $u^2 = u$ in \cite{gao2013skew}. Additionally, skew constacyclic codes over finite chain rings were explored in \cite{jitman2010skew}, along with studies on skew constacyclic codes over $\mathbb{F}_q + u\mathbb{F}_q$, where $u^2 = u$ in \cite{gao2017skew}. However, it is important to note that all these studies were conducted in the context of skew polynomial rings with automorphisms only.
\par
In 2013, Boulagouaz and Leroy in \cite{Boulagouaz} first investigated codes over a skew polynomial ring defined using both an automorphism and a derivation.  In 2014, Boucher and Ulmer in \cite{boucher2014linear} studied linear codes by using a skew polynomial ring with an automorphism as well as a derivation. In 2018, Sharma and Bhaintwal in \cite{Bhaintwal2018} studied skew cyclic codes over $\mathbb{Z}_4+u\mathbb{Z}_4, u^2=1$ determined by an automorphism and a derivation. Subsequently, many coding theorists studied codes over skew polynomial rings with automorphisms and derivations. For instance, the study was undertaken for skew cyclic codes over $\mathbb{F}_q[u, v]/\langle u^2-u, v^2-v, uv-vu\rangle$  in \cite{patel2021}, for skew constacyclic codes over $\mathbb{F}_{p^{2m}} + u\mathbb{F}_{p^{2m}}$, where $u^2 = u$ in \cite{jian2022}, and for skew constacyclic codes over $\mathbb{F}_{p^{e}}[v]/\langle v^k-1\rangle,$ where $k\mid p-1$ in \cite{Ashutosh2024}. The rings considered in the above-mentioned works are isomorphic to the product ring $\mathbb{F}_q^l,$ for some $l>1.$ For instance, $\mathbb{F}_q[u, v]/\langle u^2-u, v^2-v, uv-vu\rangle$ is isomorphic to $\mathbb{F}_q^4.$ Boucher and Ulmer in \cite{boucher2014linear} showed that in some cases, codes as modules over skew polynomial rings, whose multiplication is defined using an automorphism and a derivation, have better distance bounds than module skew codes constructed only with an automorphism. This motivates us to study skew constacyclic codes in a more general setup. In this article, we study skew constacyclic codes over the product ring $\mathcal{R}:=\mathbb{F}_q^l$ with $l\ge 1$ constructed using an automorphism and a derivation. It is also worth mentioning that the automorphism class that we consider is much larger than what was considered in some previous works; for instance, see \cite{patel2021}, \cite{jian2022}, and \cite{Ashutosh2024}. We consider a more general form of an automorphism of $\mathbb{F}_q^l,$ namely, $\theta_1\times \theta_2\times\cdots\times \theta_l,$ where each $\theta_j$ is an automorphism of $\mathbb{F}_q.$ We have shown that there exists an optimal $[12,6,6]$-linear code over $\mathbb{F}_8$ using the automorphism $\sigma\times \sigma^2$ of $\mathbb{F}_8^2$ which is not achievable using the automorphism $\sigma\times \sigma$ of $\mathbb{F}_8^2,$ where $\sigma$ is the Frobenius automorphism of $\mathbb{F}_8.$
\par
The structure of the article is as follows. Preliminaries are provided in the next section. In Section \ref{Section 3}, we study the automorphisms and derivations of the product ring $\mathcal{R}:=\mathbb{F}_q^l, l\ge1$ and give their decompositions. In Section \ref{section 4}, we introduce $(\theta, \delta_\theta, \alpha)$- cyclic codes over $\mathbb{F}_q$ and study $(\Theta, \Delta_{\Theta, \mathbf{s}}, \bm{a})$-cyclic codes over $\mathcal{R}.$ We also find the decomposition and a generator polynomial of a $(\Theta, \Delta_{\Theta, \mathbf{s}}, \bm{a})$-cyclic codes over $\mathcal{R}.$ Section \ref{Section 5} introduces Gray maps from $\mathcal{R}^n$ to $\mathbb{F}_q^{nl}.$ In Section \ref{Section 6}, we construct quantum codes from $(\Theta, \Delta_{\Theta, \mathbf{s}}, \bm{a})$-cyclic codes over $\mathcal{R}$ using the CSS construction. We also construct quantum codes by employing the annihilator dual-containing CSS construction. In Section \ref{Section 7}, we present several MDS and almost MDS linear codes and quantum codes. Section \ref{Conclusion} concludes the article.

\section{Preliminaries}\label{Preliminaries}
Throughout this article, for a ring $A$ with unity, $A^{\times}$ denotes the group of units of $A$ and $\Aut(A)$ denotes the group of all automorphisms of $A.$ For a prime power $q=p^m,$ $\mathbb{F}_q$ denotes the finite field of order $q.$
\begin{definition}
    Let $A$ be a ring with unity and let $\theta\in \Aut(A).$ An additive map $\delta_\theta: A\to A$ is called a \textit{$\theta$-derivation} of $A,$ if $\delta_\theta(r_1r_2)=\delta_\theta(r_1)r_2+\theta(r_1)\delta_\theta(r_2),$ for all $r_1,r_2\in A.$ In particular, if $\theta=Id$, then $\theta$-derivation is simply called a derivation of $A.$
\end{definition}
\begin{remark}
    For $q=p^m,$ the map $\sigma_p:\mathbb{F}_q\to \mathbb{F}_q$ given by $a\mapsto a^p$ is an automorphism of $\mathbb{F}_q$, which is called the \textit{Frobenius automorphism} of $\mathbb{F}_q$. Moreover,  $\Aut(\mathbb{F}_q)=\{\sigma_p, \sigma_p^2,\dots, \sigma_p^m=Id\}.$ 
\end{remark}

\begin{remark}
For a field $\mathbb{F}$ and $a\in \mathbb{F},$ the map $\delta_{\theta,a}:=a(\theta-Id)$ is a $\theta$-derivation of $\mathbb{F}$, which is called an \textit{inner derivation} \cite{Boulagouaz}.  It is well known that every $\theta$-derivation of $\mathbb{F}$ is an inner derivation.
\end{remark}
Let $A$ be a ring with unity, $\theta\in \Aut(A),$ and let $\delta_\theta$ be a $\theta$-derivation of $A.$ The set of all polynomials $A[x;\theta,\delta_\theta]=\{r_nx^n+r_{n-1}x^{n-1}+\cdots+r_0: r_i\in A, \; n\in \mathbb{N}\}$ over $A$ forms a non-commutative ring with respect to usual addition of polynomials and the multiplication given by $xr=\theta(r)x+\delta_\theta(r)$ for any $r \in A$ and it is extended to other elements of the ring by associativity and distributivity. We usually call this ring a \textit{skew polynomial ring}. The right division algorithm in $A[x;\theta,\delta_\theta]$ is as follows:
\begin{theorem}[Right Division Algorithm]\cite{jian2022}
    Let $f(x),g(x)\in A[x;\theta,\delta_\theta],$ where the leading coefficient of $g(x)$ is a unit. Then there exists $q(x), r(x)\in A[x;\theta,\delta_\theta]$ such that
    $$f(x)=q(x)g(x)+r(x),$$
    where $r(x)=0$ or $\deg{r}(x)<\deg{g}(x).$
\end{theorem}
If $r(x)=0,$ then $g(x)$ is called a \textit{right divisor} of $f(x).$ A left division algorithm in $A[x;\theta,\delta_\theta]$ holds and a left divisor is also defined similarly. However, in this article, we refer only to the right divisors.
\begin{definition}
    Let $A$ be a ring with unity. A \textit{pseudo-linear transformation} $T_{\theta, \delta_\theta, M}: A^n\to A^n$ is an additive map of the form
    \begin{equation}\label{Defn of semi-linear map}
        T_{\theta, \delta_\theta, M}(\bm{r})=\theta(\bm{r})M + \delta_\theta(\bm{r}),
    \end{equation}
    where $\theta\in \Aut(A),$ $\delta_\theta$ is a $\theta$-derivation of A, $\bm{r}=(r_1, r_2, \dots, r_n)\in A^n,$ $\theta(\bm{r})=(\theta(r_1), \theta(r_2), \dots, \\\theta(r_n)),$ $ \delta_\theta(\bm{r})=(\delta_\theta(r_1), \delta_\theta(r_2), \dots, \delta_\theta(r_n))$ and $M\in M_n(A).$
\end{definition}
\begin{definition}
A commutative ring $R$ with unity $1_R$ is called an \textit{$\mathbb{F}_q$-algebra} if there exists a ring homomorphism $f:\mathbb{F}_q\to R$ such that $f(1)=1_R.$
\end{definition}
\begin{example}
    Consider the $l$-dimensional $\mathbb{F}_q$-vector space $\mathbb{F}_q^l.$ It is also a ring (usually called a \textit{product ring}) with respect to the component-wise multiplication, i.e., for $\mathbf{x}=(x_1, x_2, \dots, x_l),$ $\mathbf{y}=(y_1, y_2,\dots, y_l)\in \mathbb{F}_q^l,$ $\mathbf{x} \mathbf{y}:=(x_1y_1, x_2y_2,\dots, x_ly_l).$ With this multiplication, it is a commutative ring with unity $(\underbrace{1, 1, \dots, 1}_{l\text{-times}})$. Note that for $a\in\mathbb{F}_q$, the map $a\mapsto (\underbrace{a, a, \dots, a}_{l\text{-times}})$ defines an $\mathbb{F}_q$-algebra structure on this product ring.
\end{example}
Let $R$ be a finite ring with unity. An \textit{$R$-linear code} $\mathcal{C}$ of length $n$ is a left $R$-submodule of $R^n.$ Elements of $\mathcal{C}$ are called \textit{codewords}. The \textit{(Hamming) weight} of a codeword $\bm{c}$ in $\mathcal{C},$ denoted by $w_H(\bm{c}),$ is equal to the number of non-zero coordinates of $\bm{c}=(c_0,c_1,\dots,c_{n-1}).$ The \textit{(Hamming) distance} between two codewords $\bm{c}, \bm{d} \in \mathcal{C},$ denoted by $d_H(\bm{c},\bm{d})$ is equal to $w_H(\bm{c}-\bm{d}).$ The \textit{(Hamming) distance} of a code $\mathcal{C},$ denoted by $d_H(\mathcal{C}),$ is defined as $d_H(\mathcal{C}):=\min\{d_H(\bm{c},\bm{d}):\bm{c}, \bm{d}\in \mathcal{C}\textnormal{ and }\bm{c}\neq \bm{d}\}.$ If $R$ is the finite field $\mathbb{F}_q,$ an $R$-linear code of length $n$ and dimension $k$ is a $k$-dimensional $\mathbb{F}_q$-subspace of $\mathbb{F}_q^n;$ often denoted by $[n,k]_q$-linear code. If the distance of an $[n,k]_q$-linear code is known, it is denoted by $[n,k,d]_q$-linear code. Any $[n,k,d]_q$-code satisfies $d\le n-k+1.$ An $[n, k, n-k+1]_q$-code is called a \textit{Maximum Distance Separable (MDS) code.}
    \par
    Any non-degenerate bilinear form $\langle.,.\rangle$ on $R^n$ is called an \textit{inner product} on $R^n$. If $R$ is commutative, then the \textit{Euclidean inner product} on $R^n$ is defined as $\langle \bm{c},\bm{d} \rangle_E := \underset{i=0}{\overset{n-1}{\sum}}c_id_i,$ for $\bm{c}=(c_0,c_1,\dots,c_{n-1}), \bm{d} =(d_0,d_1,\dots, d_{n-1})\in R^n$ and the \textit{Euclidean dual} of a linear code $\mathcal{C}$ over $R$ of length $n$ is defined as $\mathcal{C}^{\bot}:=\{\bm{x}\in R^n:\langle \bm{x}, \bm{c} \rangle_E =0, \, \forall \, \bm{c}\in \mathcal{C}\}.$
    \begin{definition}
        For a finite commutative ring with unity $R,$ a linear code $\mathcal{C}$ over $R$ is called \textit{Euclidean dual-containing} if $\mathcal{C}^\bot \subseteq \mathcal{C}.$
    \end{definition}
A $q$-ary \textit{quantum code} $Q$ of length $n$ and size $K$ is a $K$-dimensional subspace of the $q^n$-dimensional Hilbert space $ (\mathbb{C}^{q})^{\otimes n}= \mathbb{C}^q\otimes\cdots\otimes \mathbb{C}^q.$ If $k:=\log_q(K),$ then we use $[[n,k,d]]_q $ to denote a $q$-ary quantum code of length $n$ and minimum distance $d$. A quantum code having minimum distance $d$ can detect up to (bit flip and phase shift) $d-1$ errors and correct up to (bit flip and phase shift) $\lfloor \frac{d-1}{2} \rfloor$ errors. Any $q$-ary $[[n,k,d]]_q$ quantum code $Q$ satisfies $2d\leq n-k+2,$ which is referred to as \textit{quantum Singleton bound}. A $q$-ary $[[n,k,d]]_q$ quantum code is said to be \textit{Maximum Distance Separable (MDS)} if $2d= n-k+2$ and it is said to be \textit{almost MDS} if $2d\ge n-k.$
\section{Automorphisms and derivations of \texorpdfstring{$\mathbb{F}_q^l$}{}}\label{Section 3}
 For $l\in \mathbb{N},$ consider the product ring $\mathcal{R}:=\mathbb{F}_q^l.$ If $\theta_1, \theta_2,\dots, \theta_l \in \Aut(\mathbb{F}_q),$ then
 \begin{equation}\label{automorphisms of R}
    \begin{split}
      \theta_1\times \theta_2 \times \cdots\times \theta_l&: \mathcal{R}\to \mathcal{R} \textnormal{ given by }\\
      \theta_1\times \theta_2 \times \cdots\times \theta_l(r_1, r_2,\dots, r_l)&:=\left(\theta_1(r_1), \theta_2(r_2),\dots, \theta_l(r_l)\right)
      \end{split}
 \end{equation}
is an automorphism of $\mathcal{R}.$ We denote the set of automorphisms of $\mathcal{R}$ of the form \eqref{automorphisms of R} by $\mathcal{A},$ that is, $\mathcal{A}:=\left\{\theta_1\times \theta_2 \times \cdots\times \theta_l: \theta_1, \theta_2,\dots, \theta_l \in \Aut(\mathbb{F}_q)\right\}\subset \Aut(\mathcal{R}).$
\begin{lemma}
    Let $\Theta\in \mathcal{A}$ and $\bm{s}\in \mathcal{R}.$ Then $\Delta_{\Theta, \bm{s}}: \mathcal{R}\to \mathcal{R}$ defined by $\Delta_{\Theta, \bm{s}}(\bm{r})=\bm{s}(\Theta(\bm{r})-\bm{r})$ is a $\Theta$-derivation of $\mathcal{R}.$
\end{lemma}
\begin{proof}
Additivity of $\Delta_{\Theta, \bm{s}}$ follows from the additivity of $\Theta$. For $\bm{r_1},
\bm{r_2}\in \mathcal{R}$,
\begin{align*}
    \Delta_{\Theta, \bm{s}}(\bm{r_1}\bm{r_2})
    =&\bm{s}(\Theta(\bm{r_1}\bm{r_2})-\bm{r_1}\bm{r_2})\\
    =&\bm{s}\Theta(\bm{r_1}\bm{r_2})-\bm{s}\bm{r_1}\bm{r_2}\\
    =&\bm{s}\Theta(\bm{r_1})\Theta(\bm{r_2})-\bm{s}\Theta(\bm{r_1})\bm{r_2}+\bm{s}\Theta(\bm{r_1})\bm{r_2}-\bm{s}\bm{r_1}\bm{r_2}\\
    =&\bm{s}\Theta(\bm{r_1})(\Theta(\bm{r_2})-\bm{r_2})+\bm{s}\bm{r_2}(\Theta(\bm{r_1})-\bm{r_1})\\
    =&\bm{r_2}\Delta_{\Theta, \bm{s}}(\bm{r_1})+\Theta(\bm{r_1})\Delta_{\Theta, \bm{s}}(\bm{r_2}).
\end{align*}
Hence, $\Delta_{\Theta, \bm{s}}$ is a derivation of $\mathcal{R}$.
\end{proof}
\begin{lemma}\label{decomposition of Delta}
    If $\Theta=\theta_1\times\theta_2\times\cdots\times \theta_l\in \mathcal{A}$ and $\Delta_{\Theta, \bm{s}}$ is a $\Theta$-derivation of $\mathcal{R},$ where $\bm{s}=(s_1, s_2,\dots, s_l)\in \mathcal{R},$ then 
    \begin{equation}
        \Delta_{\Theta, \bm{s}}(\bm{r})=(\delta_{\theta_1, s_1}(r_1), \delta_{\theta_2, s_2}(r_2), \dots, \delta_{\theta_l, s_l}(r_l)), 
    \end{equation}
    for all $r=(r_1, r_2, \dots, r_l)\in \mathcal{R}.$
\end{lemma}
\begin{proof}
    Observe that
   \begin{align*}
       \Delta_{\Theta, \bm{s}}(\bm{r})&=\bm{s}(\Theta(\bm{r})-\bm{r})\\
       &= (s_1, s_2,\dots, s_l)\left((\theta_1(r_1), \theta_2(r_2), \dots, \theta_l(r_l))-(r_1, r_2, \dots, r_l)\right)\\
       &=(s_1, s_2,\dots, s_l)(\theta_1(r_1)-r_1, \theta_2(r_2)-r_2, \dots, \theta_l(r_l)-r_l)\\
       &=(s_1(\theta_1(r_1)-r_1), s_2(\theta_2(r_2)-r_2), \dots, s_l(\theta_l(r_l)-r_l))\\
       &=(\delta_{\theta_1, s_1}(r_1), \delta_{\theta_2, s_2}(r_2), \dots, \delta_{\theta_l, s_l}(r_l)).
   \end{align*}
\end{proof}
\begin{lemma}\label{thetafixess}
    Let $\Theta\in \mathcal{A}$ and $\bm{s}\in \mathcal{R}.$ If $\Theta$ fixes $\bm{s},$ then $\Delta_{\Theta, \bm{s}} \circ \Theta=\Theta \circ \Delta_{\Theta, \bm{s}}.$ 
\end{lemma}
\begin{proof}
% Let $\bm{a}=(a_1,a_2,\dots,a_l)\in\mathcal{R}$ and $\Theta=\theta_1\times\dots\times\theta_l.$ Then, 
% \begin{align*}
%     \Delta_{\Theta, \bm{s}}(\Theta(\bm{a}))=&\bm{s}(\Theta(\Theta(\bm{a}))-\Theta(\bm{a})))\\
%     =&\bm{s}\left(\Theta(\Theta(\bm{a})-\bm{a})\right)\\
%     =&\Theta(\bm{s}(\Theta(\bm{a})-\bm{a})), &&\text{since $\Theta$ fixes $\bm{s}$}\\
%     =&\Theta(\Delta_{\Theta, \bm{s}}(\bm{a}))
% \end{align*}
Immediate.
\end{proof}

\begin{corollary}
    If $\Theta\in \mathcal{A}$ and $\bm{a}\in \mathbb{F}_p^{l},$ then $\Delta_{\Theta,\bm{a}} \circ \Theta=\Theta \circ \Delta_{\Theta, \bm{a}}.$ 
\end{corollary}
\begin{proof}
    If $\Theta=\theta_1\times\theta_2\times\cdots\times \theta_l,$ then $\Theta(\bm{a})=(\theta_1(a_1), \theta_2(a_2), \dots, \theta_l(a_l))$$=(a_1, a_2,\dots, a_l)=\bm{a}.$ The proof now follows from Lemma \ref{thetafixess}.
\end{proof}
 
\section{\texorpdfstring{$(\Theta, \Delta_{\Theta,\bm{s}}, \bm{a})$}{}-cyclic codes over \texorpdfstring{$\mathcal{R}$}{}}\label{section 4}
In this section, we first introduce $(\theta, \delta_\theta, \alpha)$-cyclic code over $\mathbb{F}_q$ and then study $(\Theta, \Delta_{\Theta, \bm{s}}, \bm{a})$-cyclic code over $\mathcal{R}.$
\subsection{\texorpdfstring{$(\theta, \delta_\theta, \alpha)$}{}-cyclic codes over \texorpdfstring{$\mathbb{F}_q$}{}}
Throughout this subsection, assume that $\theta\in \Aut(\mathbb{F}_q), \delta_\theta$ be a $\theta$-derivation of $\mathbb{F}_q$ and $\alpha\in \mathbb{F}_q^{\times}.$
\begin{definition}\cite{jian2022}
    An $\mathbb{F}_q$-subspace $\mathcal{C}$ of $\mathbb{F}_q^n$ is called a $(\theta, \delta_\theta, \alpha)$-cyclic code of length $n$ over $\mathbb{F}_q$ if $T_{\theta, \delta_\theta, M_\alpha}(\mathcal{C})\subseteq \mathcal{C},$ where $T_{\theta, \delta_\theta, M_\alpha}$ is defined in Equation \eqref{Defn of semi-linear map} with 
    \begin{equation}
        M_{\alpha}=\begin{pmatrix}
            0 & 1 & 0 & \dots & 0\\
            0 & 0 & 1 & \dots & 0\\
            \vdots & \vdots & \vdots &  & \vdots\\
            0 & 0 & 0 & \dots & 1\\
            \alpha & 0 & 0 & \dots & 0\\
        \end{pmatrix}.
    \end{equation}
\end{definition}
%\begin{remark}
    %From now onwards, if $\mathcal{C}$ is a $(\theta, \delta_\theta, \alpha)$-cyclic code of length $n$ over $\mathbb{F}_q,$ we will always assume that the order of $\theta$ divides $n$ and $\alpha \in (\mathbb{F}_q^{\times})^\theta.$
%\end{remark}
For $\bm{f}(x)$ in $\mathbb{F}_q[x; \theta, \delta_\theta],$ define $\langle\bm{f(x)}\rangle=\{\bm{r}(x)\bm{f}(x): \bm{r}(x)\in \mathbb{F}_q[x; \theta, \delta_\theta]\}.$ Observe that $\langle x^n-\alpha \rangle$ may not be a two-sided ideal of $\mathbb{F}_q[x; \theta, \delta_\theta]$ and hence, $\mathbb{F}_q[x; \theta, \delta_\theta]/\langle x^n-\alpha \rangle$ may not be a ring. However, the quotient $\mathbb{F}_q[x; \theta, \delta_\theta]/\langle x^n-\alpha \rangle$ is always a left $\mathbb{F}_q[x; \theta, \delta_\theta]$-module, where $\mathbb{F}_q[x; \theta, \delta_\theta]$-multiplication is defined as
$$
    \bm{u}(x)(\bm{v}(x)+\langle x^n-\alpha \rangle)=\bm{u}(x)\bm{v}(x)+\langle x^n-\alpha \rangle,
$$
for $\bm{u}(x)\in \mathbb{F}_q[x; \theta, \delta_\theta]$ and $\bm{v}(x)+\langle x^n-\alpha \rangle \in \mathbb{F}_q[x; \theta, \delta_\theta]/\langle x^n-\alpha\rangle.$
If $\mathcal{C}$ is an $\mathbb{F}_q$-subspace of $\mathbb{F}_q^{n},$ then any $\bm{c}=(c_0, c_1, \dots, c_{n-1})\in \mathcal{C}$ can be represented uniquely by a polynomial $\bm{c}(x):=c_0+c_1x+\cdots+c_{n-1}x^{n-1}\in \mathbb{F}_q[x; \theta, \delta_\theta]/\langle x^n-\alpha\rangle.$ Identifying codewords with this polynomial representation, we have an equivalent characterization for $(\theta, \delta_\theta, \alpha)$-cyclic codes over $\mathbb{F}_q.$
\begin{proposition}\label{equivalent criterion for theta delta codes over F_q}
    An $\mathbb{F}_q$-subspace $\mathcal{C}$ of $\mathbb{F}_q^{n}$ is a $(\theta, \delta_\theta, \alpha)$-cyclic code over $\mathbb{F}_q$ if and only if $x\bm{c}(x)\in \mathcal{C}$ whenever $\bm{c}(x)\in \mathcal{C}.$
\end{proposition}
\begin{proof}
    It is not difficult to see that $x\bm{c}(x)$ represents the codeword $T_{\theta, \delta_\theta, M_\alpha}(\bm{c}).$
\end{proof}
The following theorem is immediate from Proposition \ref{equivalent criterion for theta delta codes over F_q}.
\begin{theorem}\label{left-submodule}\cite{boucher2014linear}
    Let $\mathcal{C}$ be a subset of $\mathbb{F}_q^n.$ Then $\mathcal{C}$ is a $(\theta, \delta_\theta, \alpha)$-cyclic code of length $n$ over $\mathbb{F}_q$ if and only if $\mathcal{C}$ is a left $\mathbb{F}_q[x; \theta, \delta_\theta]$-submodule of $\mathbb{F}_q[x; \theta, \delta_\theta]/\langle x^n-\alpha \rangle$.
\end{theorem}
%\begin{corollary}
%Suppose that $\circ(\theta)\mid n$ and $\alpha \in (\mathbb{F}_q^{\times})^\theta.$ Then $\mathcal{C}$ is a $(\theta, \delta_\theta, \alpha)$-cyclic code of length $n$ over $\mathbb{F}_q$ if and only if $\mathcal{C}$ is a left ideal of $\mathbb{F}_q[x; \theta, \delta_\theta]/\langle x^n-\alpha \rangle.$
%\end{corollary}
\begin{theorem}\label{generator polynomial_F_q}\cite{boucher2014linear}
    Let $\mathcal{C}$ be a $(\theta, \delta_\theta, \alpha)$-cyclic code of length $n$ over $\mathbb{F}_q$. Then $\mathcal{C}$ is generated by a unique monic polynomial $\bm{g}(x)\in \mathbb{F}_q[x; \theta, \delta_\theta].$ Moreover, $\bm{g}(x)$ is a right divisor of $x^n-\alpha$ in $\mathbb{F}_q[x; \theta, \delta_\theta].$ 
\end{theorem}
\begin{remark}
Let $\mathcal{C}$ be a $(\theta, \delta_\theta, \alpha)$-cyclic code of length $n$ over $\mathbb{F}_q.$ Then by Theorem \ref{left-submodule}, $\mathcal{C}$ is a left $\mathbb{F}_q[x; \theta, \delta_\theta]$-submodule of $\mathbb{F}_q[x; \theta, \delta_\theta]/\langle x^n-\alpha \rangle$. Consequently, by Theorem \ref{generator polynomial_F_q}, there is a unique monic polynomial $\bm{g}(x)$ of least degree in $\mathcal{C}$ that generates $\mathcal{C}.$ The polynomial $\bm{g}(x)$ is called the generator polynomial of $\mathcal{C}.$ Moreover, if $\bm{g}(x)=\sum_{i=0}^{n-k} g_ix^i,$ then a generator matrix for $\mathcal{C}$ is given by 
\begin{equation*}
    G_\mathcal{C}=
    \begin{pmatrix}
        \bm{g}\\
        T_{\theta, \delta_\theta, M_\alpha}(\bm{g})\\
        T_{\theta, \delta_\theta, M_\alpha}^{2}(\bm{g})\\
        \vdots\\
        T_{\theta, \delta_\theta, M_\alpha}^{k-1}(\bm{g})
    \end{pmatrix},
\end{equation*}
where $\bm{g}:=(g_0, g_1, \dots, g_{n-k}, 0, 0, \dots, 0)\in \mathbb{F}_q^{n}.$
\end{remark}

\subsection{\texorpdfstring{$(\Theta, \Delta_{\Theta,\bm{s}}, \bm{a})$}{}-cyclic codes over \texorpdfstring{$\mathcal{R}$}{}}\label{subsection4.2}
Recall that $\mathcal{R}=\mathbb{F}_q^{l}.$ Let $\Theta \in \mathcal{A}$ and $\Delta_{\Theta, \bm{s}}$ be a $\Theta$-derivation of $\mathcal{R}.$ Choose $\bm{a}\in \mathcal{R}^{\times}.$
%\begin{definition}
    %{\color{blue} A subset $\mathcal{C}$ of $\mathcal{R}^{n}$ is said to be a $(\Theta,\Delta_{\Theta, \bm{s}})$-linear code if it is a left $\mathcal{R}[x;\Theta,\Delta_{\Theta, \bm{s}}]$-submodule of $\mathcal{R}[x;\Theta,\Delta_\Theta]/\langle f(x) \rangle$ where $f(x)$ is an arbitrary polynomial of degree of $n$ over $\mathcal{R}$.}
%\end{definition}
\begin{definition}\label{definition of cyclic over R}
    An $\mathcal{R}$-submodule $\mathcal{C}$ of $\mathcal{R}^n$ is called a $(\Theta, \Delta_{\Theta, \bm{s}}, \bm{a})$-cyclic code of length $n$ over $\mathcal{R}$ if $T_{\Theta, \Delta_{\Theta, \bm{s}}, M_{\bm{a}}}(\mathcal{C})\subseteq \mathcal{C},$ where $T_{\Theta, \Delta_{\Theta, \bm{s}}, M_{\bm{a}}}$ is defined in Equation \eqref{Defn of semi-linear map} with 
    \begin{equation}\label{Matrix over R}
        M_{\bm{a}}=\begin{pmatrix}
            \bm{0} & \bm{1} & \bm{0} & \dots & \bm{0}\\
            \bm{0} & \bm{0} & \bm{1} & \dots & \bm{0}\\
            \vdots & \vdots & \vdots &  & \vdots\\
            \bm{0} & \bm{0} & \bm{0} & \dots & \bm{1}\\
            \bm{a} & \bm{0} & \bm{0} & \dots & \bm{0}\\
        \end{pmatrix},
    \end{equation}
    where $\bm{a}=(a_1, a_2, \dots, a_l), \bm{1}=(1, 1, \dots, 1)$ and $\bm{0}=(0, 0, \dots, 0) \in \mathcal{R}.$
\end{definition}
%\begin{remark}
    %From now onwards, if $\mathcal{C}$ is a $(\Theta, \Delta_{\Theta, \bm{s}}, \bm{a})$-cyclic code of length $n$ over $\mathcal{R},$ we will always assume that $\bm{a}\in \mathcal{R}\setminus\{\bm{0}\}$ with $\Theta(\bm{a})=\bm{a}$ and $\circ(\Theta) \mid n.$
%\end{remark}
The quotient $\mathcal{R}_n:=\mathcal{R}[x; \Theta, \Delta_{\Theta, \bm{s}}]/\langle x^n-\bm{a}\rangle$ is a left $\mathcal{R}[x; \Theta, \Delta_{\Theta, \bm{s}}]$-module, where $\mathcal{R}[x; \Theta, \Delta_{\Theta, \bm{s}}]$-multiplication is defined as
$$
    \Bar{\bm{s}}(x)(\Bar{\bm{r}}(x)+\langle x^n-\bm{a} \rangle)=\Bar{\bm{s}}(x)\Bar{\bm{r}}(x)+\langle x^n-\bm{a} \rangle,
$$
where $\Bar{\bm{s}}(x)\in \mathcal{R}[x; \Theta, \Delta_{\Theta, \bm{s}}]$ and $\Bar{\bm{r}}(x)+\langle x^n-\bm{a} \rangle\in \mathcal{R}_n.$ Consider the $\mathcal{R}$-module isomorphism
\begin{equation}\label{identification}
\begin{split}
    \eta &: \mathcal{R}^{n}\to \mathcal{R}[x; \Theta, \Delta_{\Theta, \bm{s}}]/\langle x^n-\bm{a}\rangle\\
    \eta(\Bar{\bm{c}}=(\bm{c}_0, \bm{c}_1, \dots, \bm{c}_{n-1}))&= \Bar{\bm{c}}(x)=\bm{c}_0+\bm{c}_1x+ \cdots+ \bm{c}_{n-1}x^{n-1}.
    \end{split}
\end{equation}
With the identification of vectors in $\mathcal{R}^n$ with elements in  $\mathcal{R}[x; \Theta, \Delta_{\Theta, \bm{s}}]/\langle x^n-\bm{a}\rangle$ via \eqref{identification}, we identify a $(\Theta, \Delta_{\Theta, \bm{s}}, \bm{a})$-cyclic code $\mathcal{C}$ over $\mathcal{R}$ by $\eta(\mathcal{C}).$
\begin{proposition} \label{equivalent criterion for theta delta codes over R}
    An $\mathcal{R}$-submodule $\mathcal{C}$ of $\mathcal{R}^{n}$ is a $(\Theta, \Delta_{\Theta, \bm{s}}, \bm{a})$-cyclic code over $\mathcal{R}$ if and only if $x\Bar{\bm{c}}(x)\in \mathcal{C}$ whenever $\Bar{\bm{c}}(x)\in \mathcal{C}.$
\end{proposition}
\begin{proof}
Let $\Bar{\bm{c}}=(\bm{c}_0, \bm{c}_1, \dots, \bm{c}_{n-1})\in \mathcal{C}.$ Then
\begin{align*}
    T_{\Theta, \Delta_{\Theta, \bm{s}}, M_{\bm{a}}}(\Bar{\bm{c}})&=\Theta(\Bar{\bm{c}})M_{\bm{a}}+\Delta_{\Theta, \bm{s}}(\Bar{\bm{c}})\\
    &=(\bm{a}\Theta(\bm{c}_{n-1})+\Delta_{\Theta, \bm{s}}(\bm{c}_0), \Theta(\bm{c}_0)+\Delta_{\Theta, \bm{s}}(\bm{c}_1), \dots, \Theta(\bm{c}_{n-2})+\Delta_{\Theta, \bm{s}}(\bm{c}_{n-1})),
\end{align*}
and 
\begin{align*}
    x\Bar{\bm{c}}(x)&=x\bm{c}_0+(x\bm{c}_1)x+\cdots+ (x\bm{c}_{n-1})x^{n-1}\\
    &=(\Theta(\bm{c}_0)x+\Delta_{\Theta, \bm{s}}(\bm{c}_{0}))+(\Theta(\bm{c}_1)x^2+\Delta_{\Theta, \bm{s}}(\bm{c}_{1})x)+\cdots+ (\Theta(\bm{c}_{n-1})x^n+\Delta_{\Theta, \bm{s}}(\bm{c}_{n-1})x^{n-1})\\
    &=(\Theta(\bm{c}_0)x+\Delta_{\Theta, \bm{s}}(\bm{c}_{0}))+(\Theta(\bm{c}_1)x^2+\Delta_{\Theta, \bm{s}}(\bm{c}_{1})x)+\cdots+ (\Theta(\bm{c}_{n-1})\bm{a}+\Delta_{\Theta, \bm{s}}(\bm{c}_{n-1})x^{n-1})\\
    &=(\bm{a}\Theta(\bm{c}_{n-1})+\Delta_{\Theta, \bm{s}}(\bm{c}_0))+(\Theta(\bm{c}_0)+\Delta_{\Theta, \bm{s}}(\bm{c}_1))x+\cdots+ (\Theta(\bm{c}_{n-2})+\Delta_{\Theta, \bm{s}}(\bm{c}_{n-1}))x^{n-1}\\
    &=\eta(T_{\Theta, \Delta_{\Theta, \bm{s}}, M_{\bm{a}}}(\Bar{\bm{c}})).
\end{align*}
Thus, with the identification given by \eqref{identification}, $T_{\Theta, \Delta_{\Theta, \bm{s}}, M_{\bm{a}}}(\Bar{\bm{c}})$ represents the codeword $ x\Bar{\bm{c}}(x).$ The proof now immediately follows from Definition \ref{definition of cyclic over R}.
\end{proof}
The following theorem is an immediate consequence of Proposition \ref{equivalent criterion for theta delta codes over R}.
\begin{theorem}\label{left-submodule R}
    Let $\mathcal{C}$ be a subset of $\mathcal{R}^n.$ Then $\mathcal{C}$ is a $(\Theta, \Delta_{\Theta,\bm{s}} \bm{a})$-cyclic code of length $n$ over $\mathcal{R}$ if and only if $\mathcal{C}$ is a left $\mathcal{R}[x; \Theta, \Delta_{\Theta, \bm{s}}]$-submodule of $\mathcal{R}[x; \Theta, \Delta_{\Theta, \bm{s}}]/\langle x^n-\bm{a}\rangle.$
\end{theorem}
\begin{proof}
    Let $\mathcal{C}$ be a $(\Theta, \Delta_{\Theta,\bm{s}}, \bm{a})$-cyclic code of length $n$ over $\mathcal{R}.$ Since $\eta$ is an isomorphism, for $\Bar{\bm{c}}(x),\,\Bar{\bm{d}}(x)\in \mathcal{C}, \Bar{\bm{c}}(x)+\Bar{\bm{d}}(x)\in \mathcal{C}.$ By the repeated application of Proposition \ref{equivalent criterion for theta delta codes over R}, for all $i\in \mathbb{N},$ $x^i\Bar{\bm{c}}(x)\in \mathcal{C}$ whenever $\Bar{\bm{c}}(x)\in \mathcal{C}.$ Thus, for any $\Bar{\bm{r}}(x)\in \mathcal{R}[x; \Theta, \Delta_{\Theta, \bm{s}}]$ and $\Bar{\bm{c}}(x)\in \mathcal{C},$  $\Bar{\bm{r}}(x)\Bar{\bm{c}}(x)\in \mathcal{C}.$ Hence, $\mathcal{C}$ is a left $\mathcal{R}[x; \Theta, \Delta_{\Theta, \bm{s}}]$-submodule of $\mathcal{R}[x; \Theta, \Delta_{\Theta, \bm{s}}]/\langle x^n-\bm{a}\rangle.$
    \par
    The converse is immediate.
\end{proof}
\begin{remark}
    Proofs of Proposition \ref{equivalent criterion for theta delta codes over R} and Theorem \ref{left-submodule R} can be found in \cite{bajalan2023sigma}. However, for readability and completion of the article, we have included the proofs.
\end{remark}
Let $\mathcal{B}:=\{\bm{e}_1< \bm{e}_2<\dots< \bm{e}_l\}$ be the standard ordered $\mathbb{F}_q$-basis for $\mathcal{R}.$ For $1\le i\le l,$ denote the canonical $i$-th projection of $\mathcal{R}$ to $\mathbb{F}_q$ by $\pi_i,$ that is $\pi_i: \mathcal{R}\to \mathbb{F}_q$ is defined as $\pi_i(\bm{r}=(r_1, r_2, \dots, r_l))=r_i.$ Extend each $\pi_i$ to $\mathcal{R}^n$ as follows:
\begin{equation}
    \begin{split}
        \Tilde{\pi}_i&: \mathcal{R}^n\to \mathbb{F}_q^{n}\\
        \Tilde{\pi}_i(\bm{r}_0,\bm{r}_1, \dots, \bm{r}_{n-1})&:=(\pi_i(\bm{r}_0), \pi_i(\bm{r}_1),\dots, \pi_i(\bm{r}_{n-1})).
    \end{split}
\end{equation}
Then, every linear code over $\mathcal{R}$ can be uniquely decomposed into a direct sum of linear codes over $\mathbb{F}_q$ as follows:
\par
Let $\mathcal{C}$ be a linear code of length $n$ over $\mathcal{R}.$ For $1\le i\le l,$ define
\begin{equation}
    \mathcal{C}_i:=\Tilde{\pi}_i(\mathcal{C})=\{(\pi_i(\bm{c}_0), \pi_i(\bm{c}_1), \dots , \pi_i(\bm{c}_{n-1}))\in \mathbb{F}_q^n: (\bm{c}_0, \bm{c}_1, \dots, \bm{c}_{n-1})\in \mathcal{C}\}.
\end{equation}
Then for each $i,$ $\mathcal{C}_i$ is a linear codes of length $n$ over $\mathbb{F}_q$ and $\mathcal{C}$ can be uniquely decomposed as
\begin{equation}
    \mathcal{C}=\bm{e}_1\mathcal{C}_1\oplus\bm{e}_2\mathcal{C}_2\oplus\cdots\oplus\bm{e}_l\mathcal{C}_l.
\end{equation}

\begin{theorem}\label{main theorem}
    Let $\mathcal{C}=\bm{e}_1\mathcal{C}_1\oplus\bm{e}_2\mathcal{C}_2\oplus\cdots\oplus\bm{e}_l\mathcal{C}_l$ be a linear code of length $n$ over $\mathcal{R},$ where $\mathcal{C}_i$ is a linear code of length $n$ over $\mathbb{F}_q,$ for $1\le i\le l.$ Then $\mathcal{C}$ is a $(\Theta, \Delta_{\Theta, \bm{s}}, \bm{a})$-cyclic code over $\mathcal{R}$ if and only if for each $1\le i\le l,\, \mathcal{C}_i$ is a $(\theta_i, \delta_{\theta_i, s_i}, a_i)$-cyclic code over $\mathbb{F}_q,$ where $\Theta=\theta_1\times \theta_2\times\cdots\times \theta_l \in \mathcal{A},$ $\bm{s}=(s_1, s_2, \dots, s_l)\in \mathcal{R}$ and $\bm{a}=(a_1, a_2, \dots, a_l)\in\mathcal{R}^{\times}.$  
\end{theorem}
\begin{proof}
    Suppose  $(\pi_i(\bm{c}_0), \pi_i(\bm{c}_1), \dots , \pi_i(\bm{c}_{n-1}))\in \mathcal{C}_i, $ where $\Bar{\bm{c}}=(\bm{c}_0, \bm{c}_1, \dots, \bm{c}_{n-1})\in \mathcal{C}.$ Observe
    \begin{align*}
            T_{\Theta, \Delta_{\Theta, \bm{s}}, M_{\bm{a}}}(\Bar{\bm{c}})&= \Theta(\Bar{\bm{c}})M_{\bm{a}}+\Delta_{\Theta, \bm{s}}(\Bar{\bm{c}})\\
            &=(\bm{a}\Theta(\bm{c}_{n-1})+\Delta_{\Theta, \bm{s}}(\bm{c}_0), \Theta(\bm{c}_0)+\Delta_{\Theta, \bm{s}}(\bm{c}_1), \dots, \Theta(\bm{c}_{n-2})+\Delta_{\Theta, \bm{s}}(\bm{c}_{n-1})).
    \end{align*}
    For $0\le j\le n-1,$ let $\bm{c}_j=(c_j^{(1)}, c_j^{(2)}, \dots, c_j^{(l)}).$ Then
    \begin{align*}
        &T_{\Theta, \Delta_{\Theta, \bm{s}}, M_{\bm{a}}}(\Bar{\bm{c}})\\
        &\begin{multlined}[t]
            =\Theta(\Bar{\bm{c}})M_{\bm{a}}+\Delta_{\Theta, \bm{s}}(\Bar{\bm{c}})
        \end{multlined}\\
        &\begin{multlined}[t]
            =\biggl((a_1,\dots, a_l)\left(\theta_1(c_{n-1}^{(1)}), \theta_2(c_{n-1}^{(2)}), \dots, \theta_l(c_{n-1}^{(l)})\right)+
            \left(\delta_{\theta_1, s_1}(c_0^{(1)}), \delta_{\theta_2, s_2}(c_0^{(2)}), \dots, \delta_{\theta_l, s_l}(c_0^{(l)})\right),\\
            \left(\theta_1(c_0^{(1)}), \theta_2(c_0^{(2)}), \dots, \theta_l(c_0^{(l)})\right)+
            \left(\delta_{\theta_1, s_1}(c_1^{(1)}), \delta_{\theta_2, s_2}(c_1^{(2)}), \dots, \delta_{\theta_l, s_l}(c_1^{(l)})\right), \dots,\\
            \left(\theta_1(c_{n-2}^{(1)}), \theta_2(c_{n-2}^{(2)}), \dots, \theta_l(c_{n-2}^{(l)})\right)+
            \left(\delta_{\theta_1, s_1}(c_{n-1}^{(1)}), \delta_{\theta_2, s_2}(c_{n-1}^{(2)}), \dots, \delta_{\theta_l, s_l}(c_{n-1}^{(l)})\right)\biggr)
        \end{multlined}\\
        &\begin{multlined}[t]
            =\biggl(\left(a_1\theta_1(c_{n-1}^{(1)})+\delta_{\theta_1, s_1}(c_0^{(1)}), a_2\theta_2(c_{n-1}^{(2)})+\delta_{\theta_2, s_2}(c_0^{(2)}), \dots, a_l\theta_l(c_{n-1}^{(l)})+\delta_{\theta_l, s_l}(c_0^{(l)})\right),\\
        \left(\theta_1(c_0^{(1)})+\delta_{\theta_1, s_1}(c_1^{(1)}), \theta_2(c_0^{(2)})+\delta_{\theta_2, s_2}(c_1^{(2)}),\dots, \theta_l(c_0^{(l)})+\delta_{\theta_l, s_l}(c_1^{(l)})\right), \dots,\\
        \left(\theta_1(c_{n-2}^{(1)})+\delta_{\theta_1, s_1}(c_{n-1}^{(1)}), \theta_2(c_{n-2}^{(2)})+\delta_{\theta_2, s_2}(c_{n-1}^{(2)}),\dots,\theta_l(c_{n-2}^{(l)})+\delta_{\theta_l, s_l}(c_{n-1}^{(l)})\right)\biggr).\end{multlined}
    \end{align*}
    Note that for each $1\le i\le l,$
    \begin{align*}
        &\left(a_i\theta_i(c_{n-1}^{(i)})+\delta_{\theta_i, s_i}(c_0^{(i)}), \theta_i(c_0^{(i)})+\delta_{\theta_i, s_i}(c_1^{(i)}), \dots,\theta_i(c_{n-2}^{(i)})+\delta_{\theta_i, s_i}(c_{n-1}^{(i)})\right)\\
        &=T_{\theta_i, \delta_{\theta_i, s_i}, M_{a_i}}(c_0^{(i)},c_1^{(i)}, \dots, c_{n-1}^{(i)} ),\,\,\text{where}
    \end{align*} 
    \begin{equation*}
        M_{a_i}=\begin{pmatrix}
            0 & 1 & 0 & \dots & 0\\
            0 & 0 & 1 & \dots & 0\\
            \vdots & \vdots & \vdots &  & \vdots\\
            0 & 0 & 0 & \dots & 1\\
            a_i & 0 & 0 & \dots & 0\\
        \end{pmatrix}.
    \end{equation*}
  Consequently, $T_{\Theta, \Delta_{\Theta, \bm{s}}, M_{\bm{a}}}(\Bar{\bm{c}})\in \mathcal{C}$ if and only if $T_{\theta_i, \delta_{\theta_i, s_i}, M_{a_i}}(c_0^{(i)},c_1^{(i)}, \dots, c_{n-1}^{(i)} )\in\mathcal{C}_i$ for each $1\leq i\leq l.$
\end{proof}
\begin{theorem}
    Let $\mathcal{C}=\bm{e}_1\mathcal{C}_1\oplus\bm{e}_2\mathcal{C}_2\oplus\cdots\oplus\bm{e}_l\mathcal{C}_l$ be a $(\Theta, \Delta_{\Theta, \bm{s}}, \bm{a})$-cyclic code of length $n$ over $\mathcal{R},$ where $\Theta=\theta_1\times \theta_2\times\cdots\times \theta_l \in \mathcal{A},$ $\bm{s}=(s_1, s_2, \dots, s_l)\in\mathcal{R}$ and $\bm{a}=(a_1, a_2, \dots, a_l)\in\mathcal{R}^{\times}.$ Then there exist $\bm{g}_1(x), \bm{g}_2(x), \dots, \bm{g}_l(x)\in \mathbb{F}_q[x]$ such that $\mathcal{C}=\langle \Bar{\bm{g}}(x) \rangle,$ where $\Bar{\bm{g}}(x)=(\bm{g}_1(x), \bm{g}_2(x), \dots, \bm{g}_l(x)).$ Additionally, $\Bar{\bm{g}}(x)$ is a right divisor of $x^n-\bm{a}$ in $\mathcal{R}[x; \Theta, \Delta_{\Theta, \bm{s}}]$ and $ |\mathcal{C}|=q^{nl-\underset{i=1}{\overset{l}{\sum}}\deg \bm{g}_i(x)}.$
\end{theorem}
\begin{proof}
    Let $\mathcal{C}=\bm{e}_1\mathcal{C}_1\oplus\bm{e}_2\mathcal{C}_2\oplus\cdots\oplus\bm{e}_l\mathcal{C}_l$ be a $(\Theta, \Delta_{\Theta, \bm{s}}, \bm{a})$-cyclic code of length $n$ over $\mathcal{R}.$ Then by Theorem \ref{main theorem}, each $\mathcal{C}_i$ is a $(\theta_i, \delta_{\theta_i, s_i}, a_i)$-cyclic code of length $n$ over $\mathbb{F}_q.$ Hence, for each $1\le i\le l,$ there exists $\bm{g}_i(x)$ such that $\mathcal{C}_i=\langle \bm{g}_i(x) \rangle\subseteq \mathbb{F}_q[x]/\langle x^n-a_i \rangle.$ Since $\bm{g}_i(x)$ is a right divisor of $x^n-a_i,$ for each $1\leq i \leq l,$ write $x^n-a_i=\bm{q}_i(x)\bm{g}_i(x),$ where $\bm{q}_i[x]\in \mathbb{F}_q[x;\theta_i, \delta_{\theta_i}].$ Let $\Bar{\bm{g}}(x)=\sum_{i=1}^l \bm{e}_i\bm{g}_i(x)$ and $\Bar{\bm{q}}(x)=\sum_{i=1}^l \bm{e}_i\bm{q}_i(x).$ Then $x^n-\bm{a}=\Bar{\bm{q}}(x)\Bar{\bm{g}}(x),$ showing that $\Bar{\bm{g}}(x)$ is a right divisor of $x^n-\bm{a}.$ Also, $|\mathcal{C}|=|\mathcal{C}_1|\times|\mathcal{C}_2|\times\cdots\times|\mathcal{C}_l|=q^{n-\deg \bm{g}_1(x)}\times q^{n-\deg \bm{g}_2(x)}\times \cdots \times q^{n- \deg \bm{g}_l(x) }=q^{nl-\underset{i=1}{\overset{l}{\sum}}\deg \bm{g}_i(x)}.$ Furthermore, if $\Bar{\bm{c}}(x)\in \mathcal{C}=\bm{e}_1\mathcal{C}_1\oplus\bm{e}_2\mathcal{C}_2\oplus\cdots\oplus\bm{e}_l\mathcal{C}_l,$ then $\Bar{\bm{c}}(x)=(\bm{c}_1(x),\bm{c}_2(x),\dots,\bm{c}_l(x)),$ where each $\bm{c}_i(x) \in \mathbb{F}_q[x]$ for each $1\leq i \leq l.$ Note that, $\Bar{\bm{c}}(x)=(\bm{c}_1(x),\bm{c}_2(x),\dots,\bm{c}_l(x))=(\bm{h}_1(x)\bm{g}_1(x),\bm{h}_2(x)\bm{g}_2(x),\dots,\bm{h}_l(x)\bm{g}_l(x)),$ where $\bm{h}_i(x)\in \mathbb{F}_q[x;\theta_i, \delta_{\theta_i}]$ for each $1\leq i \leq l.$ Hence, if $\Bar{\bm{h}}(x)=(\bm{h}_1(x), \bm{h}_2(x),\dots, \bm{h}_l(x)),$ then $\Bar{\bm{c}}(x)=\Bar{\bm{h}}(x)\Bar{\bm{g}}(x).$ Consequently, $\mathcal{C}\subseteq \langle (\bm{g}_1(x), \bm{g}_2(x), \dots, \bm{g}_l(x)) \rangle.$ Conversely, since for $1\leq i \leq l,$\,$\langle (\bm{0},\dots,\bm{0},\bm{g}_i(x),\bm{0},\dots,\bm{0}) \rangle\subseteq \mathcal{C},$  we must have $\langle (\bm{g}_1(x), \bm{g}_2(x), \dots, \bm{g}_l(x) )\rangle\subseteq \mathcal{C}.$ This completes the proof.
    \end{proof}
\begin{remark}
   All Propositions and Theorems stated in Subsection \ref{subsection4.2} remain valid even when $\mathcal{R}=\mathbb{F}_{q_1}\times\mathbb{F}_{q_2}\times\dots\times\mathbb{F}_{q_t},$ where $q_i$ is a prime power (not necessarily a power of the same prime) for $1\le i\le t$. However, the aim of this article is to construct codes with good parameters over a finite field by using a Gray map. If we assume that $q_1=q_2=\dots=q_t=q,$ then by applying a Gray map on codes over $\mathcal{R}$ (see Section \ref{Section 5}), we get $\mathbb{F}_q$-linear codes. For $\mathcal{R}=\mathbb{F}_{q_1}\times\mathbb{F}_{q_2}\times\dots\times\mathbb{F}_{q_t},$ one can construct additive codes by defining a Gray map.
\end{remark}
\section{Gray maps on $\mathcal{R}$}\label{Section 5}
In this article, we will consider the following definition of a Gray map. 
\begin{definition}\label{defn of Gray map}
    A \textit{Gray map} from a finite-dimensional $\mathbb{F}_q$-algebra $A$ of dimension $t$ to $\mathbb{F}_q^t$ is an $\mathbb{F}_q$-isomorphism between the $\mathbb{F}_q$-vector spaces $A$ and $\mathbb{F}_q^t$ .
\end{definition}
Let $\GL(l, \mathbb{F}_q)$ denote the set of all $l\times l$ invertible matrices with entries from $\mathbb{F}_q.$ Define the map $\Phi$ as follows:
\begin{equation}
\begin{split}
    \Phi &: \mathcal{R}^n\rightarrow \mathbb{F}_q^{nl}\\
    \Phi(\bm{a}_0,\bm{a}_1,\dots, \bm{a}_{n-1})&=(\bm{a}_0M_0,\bm{a}_1M_1,\dots, \bm{a}_{n-1}M_{n-1}),
\end{split}
\end{equation}
where $\bm{a}_i=(a_i^{(1)}, a_i^{(2)}, \dots, a_i^{(l)})$ and $M_i\in \GL(l, \mathbb{F}_q),$ for $0\le i\le n-1.$ 
\begin{theorem}
The map $\Phi$ is a Gray map.
\end{theorem}
\begin{proof}
   It is not difficult to prove that $\Phi$ is $\mathbb{F}_q$-linear. Since $M_0,M_1,\dots, M_{n-1}\in \GL(l, \mathbb{F}_q), \Phi$ is invertible. Hence, $\Phi$ is an $\mathbb{F}_q$-isomorphism from $\mathcal{R}^n$ to $\mathbb{F}_q^{nl}.$
\end{proof}
Define the \textit{Gray weight} of an element $(\bm{a}_0,\bm{a}_1,\dots,\bm{a}_{n-1})$ in $\mathcal{R}^n$ by
\begin{equation}
    w_G(\bm{a}_0,\bm{a}_1,\dots,\bm{a}_{n-1}):=\underset{i=0}{\overset{n-1}{\sum}}w_H(\bm{a}_iM_i),
\end{equation}
where $w_H(\bm{a}_iM_i)$ denotes the Hamming weight of $\bm{a}_iM_i.$ The \textit{Gray distance} between $\Bar{\bm{a}}, \Bar{\bm{b}} \in \mathcal{R}^n,$ denoted by $d_G(\Bar{\bm{a}},\Bar{\bm{b}}),$ is $w_G(\Bar{\bm{a}}-\Bar{\bm{b}})$ and the \textit{(minimum) Gray distance} of the linear code $\mathcal{C}$ over $\mathcal{R},$ denoted by $d_G(\mathcal{C}),$ is $\min \{d_G(\Bar{\bm{a}},\Bar{\bm{b}}) \,|\, \Bar{\bm{a}}, \Bar{\bm{b}} \in \mathcal{C} \textnormal{ and }\Bar{\bm{a}}\neq \Bar{\bm{b}}\}.$ 
\begin{theorem}\label{Graymap}
     The Gray map $\Phi$ is a distance preserving map from $\mathcal{R}^n$ (Gray distance) to $\mathbb{F}_q^{nl}$ (Hamming distance). Moreover, if $\mathcal{C}$ is an $\mathcal{R}$-linear code of length $n$ with Gray distance $d_G,$ then $\Phi(\mathcal{C})$ is an $\mathbb{F}_q$-linear code with parameters $[nl, k, d_G]_q,$ where $k=\dim_{\mathbb{F}_q}\mathcal{C}.$
\end{theorem}
\begin{proof}
    It is immediate from the definition of Gray distance that $\Phi$ is distance preserving. If $k=\dim_{\mathbb{F}_q}\mathcal{C},$ then the cardinality of $\mathcal{C}=q^k=|\Phi(\mathcal{C})|.$ Hence, $\dim_{\mathbb{F}_q}(\Phi(\mathcal{C}))=k.$  The length and minimum distance of $\Phi(\mathcal{C})$ are evident from the definition.
\end{proof}
The map $\Phi$ preserves orthogonality for certain types of matrices in $\GL(l, \mathbb{F}_q)$ that define $\Phi$ as described in the following theorem. 
\begin{theorem}\label{Phi preserves orthogonality}
    Let $\mathcal{C}$ be a Euclidean dual-containing linear code over $\mathcal{R}$. Then $\Phi(\mathcal{C})$ is Euclidean dual-containing if $M_iM_i^T=\lambda I_l,$ for each $0 \leq i \leq n-1,$ where $\lambda\in \mathbb{F}_q^{\times}.$
\end{theorem}
\begin{proof}
Let $\Bar{\bm{c}}=(\bm{c}_0,\bm{c}_1,\dots,\bm{c}_{n-1})\in \mathcal{C}$ and $\Bar{\bm{d}}=(\bm{d}_0,\bm{d}_1,\dots,\bm{d}_{n-1})\in \mathcal{C}^\bot.$ Further assume that, $\bm{c}_i=(c_{i}^{(1)},c_{i}^{(2)},\dots,c_{i}^{(l)})$ and $\bm{d}_i=(d_{i}^{(1)},d_{i}^{(2)},\dots,d_{i}^{(l)})$ for each $0 \leq i \leq {n-1}.$ Then, $\langle\Bar{\bm{c}}, \Bar{\bm{d}}\rangle_E=\underset{j=0}{\overset{n-1}{\sum}} \bm{c}_j\bm{d}_j.$ Now,
$$
    \langle\Phi(\Bar{\bm{c}}), \Phi( \Bar{\bm{d}})\rangle_E
    =\underset{i=0}{\overset{n-1}{\sum}}(\bm{c}_iM_i)(\bm{d}_iM_i)^T
    = \underset{i=0}{\overset{n-1}{\sum}}\bm{c}_iM_iM_i^T\bm{d}_i^T
    =\lambda\underset{i=0}{\overset{n-1}{\sum}}\underset{j=1}{\overset{l}{\sum}}c_i^{(j)}d_i^{(j)}.
$$
Since $\langle\Bar{\bm{c}}, \Bar{\bm{d}}\rangle_E=\Bar{\bm{0}},\,\,\underset{i=0}{\overset{n-1}{\sum}}c_i^{(j)}d_i^{(j)}=0,$ for each $1\le j\le l.$ 
Consequently, $\langle\Phi(\Bar{\bm{c}}),\Phi(\Bar{\bm{d}})\rangle_E=\bm{0}.$
\par
Since $\mathcal{C}^{\bot}\subseteq \mathcal{C},$ $\Phi(C^{\bot})\subseteq \Phi(\mathcal{C}).$ Thus, to show that $\Phi(C)$ is Euclidean dual-containing, it is enough to show that $\Phi(\mathcal{C}^{\bot})=\Phi(\mathcal{C})^{\bot}.$
Let $\Phi(\bm{d})\in \Phi(\mathcal{C}^{\bot}), $ where $\Bar{\bm{d}}=(\bm{d}_0,\bm{d}_1,\dots,\bm{d}_{n-1})\in \mathcal{C}^\bot$ and let $\Phi(\bm{c})\in \Phi(\mathcal{C}),$ where $\Bar{\bm{c}}=(\bm{c}_0,\bm{c}_1,\dots,\bm{c}_{n-1})\in \mathcal{C}.$  Since $\langle \Bar{\bm{c}}, \Bar{\bm{d}} \rangle_E = \Bar{\bm{0}}, \forall \Bar{\bm{c}}\in\mathcal{C},$ we have $\langle \Phi(\Bar{\bm{d}}), \Phi(\Bar{\bm{c}}) \rangle_E = \bm{0},\forall \Phi(\Bar{\bm{c}})\in\Phi(\mathcal{C}).$ Consequently, $\Phi(\Bar{\bm{d}} )\in \Phi(\mathcal{C})^\perp$ and therefore $\Phi(\mathcal{C}^{\bot})\subseteq\Phi(\mathcal{C})^{\bot}.$ Since $\Phi$ is an isomorphism, $|\Phi(\mathcal{C}^{\bot})|=|\Phi(\mathcal{C})^{\bot}|$ and therefore, $\Phi(\mathcal{C}^{\bot})=\Phi(\mathcal{C})^{\bot}.$
\end{proof}
\section{Quantum codes from \texorpdfstring{$(\Theta, \bm{0}, \bm{a})$}{}-cyclic codes over \texorpdfstring{$\mathcal{R}$}{}}\label{Section 6}
If $\mathcal{C}=\langle \bm{g}(x)\rangle$ is a $(\theta, \delta, \alpha)$-cyclic code over $\mathbb{F}_q,$  then it is not true in general that $\mathcal{C}^{\bot}=\langle \bm{h}(x) \rangle,$ where $x^n-\alpha=\bm{g}(x)\bm{h}(x)=\bm{h}(x)\bm{g}(x).$ In this section, $q$-ary quantum codes are constructed from $(\Theta, \mathbf{0}, \bm{a})$-cyclic codes over $\mathcal{R}.$
%Suppose $\mathcal{C}=\langle \bm{g}(x)\rangle$ is a $(\theta, \delta_\theta, \alpha)$-cyclic code of length $n$ over $\mathbb{F}_q,$ for some right divisor $\bm{g}(x)$ of $x^n-\alpha.$ If $\bm{g}(x)$ is also a left divisor of $x^n-\alpha,$ then we determine a necessary and sufficient condition for $\mathcal{C}^{\bot}\subseteq \mathcal{C}.$ 
\subsection{Dual-containing codes with respect to the Euclidean inner product}
\begin{lemma}\label{center}\cite{BOUCHER2009}
    If $\bm{h}(x)\bm{g}(x)\in Z(\mathbb{F}_q[x; \theta, 0]),$ then $\bm{g}(x)\bm{h}(x)=\bm{h}(x)\bm{g}(x)$ in $\mathbb{F}_q[x; \theta, 0].$
\end{lemma}

\begin{theorem}
    Let $\mathcal{C}$ be a $(\theta, 0, \alpha)$-cyclic code of length $n$ over $\mathbb{F}_q,$ where $\alpha \in (\mathbb{F}_q^{\times})^{\theta}:=\{v\in \mathbb{F}_q*: \theta(v)=v\}$ and the order of $\theta$ divides $n.$ Then $\mathcal{C}^{\bot}$ is a $(\theta, 0, \alpha^{-1})$-cyclic code of length $n$ over $\mathbb{F}_q.$
\end{theorem}
\begin{proof}
    Let $\bm{c}=(c_0, c_1, \dots, c_{n-1})\in \mathcal{C}.$ Then
    \begin{equation}
        T_{\theta, 0, M_{\alpha}}^{n-1}(\bm{c})=\left(\alpha \theta^{n-1}(c_1), \alpha \theta^{n-1}(c_2), \dots, \alpha \theta^{n-1}(c_{n-1}), \theta^{n-1}(c_0)\right)\in \mathcal{C}.
    \end{equation}
    Let $\bm{d}=(d_0, d_1, \dots, d_{n-1}) \in \mathcal{C}^{\bot}.$ Then
    \begin{align*}
        0&=\langle T_{\theta, 0, M_{\alpha}}^{n-1}(\bm{c}),\; \bm{d}\rangle_E\\
        &= \langle \left(\alpha \theta^{n-1}(c_1), \alpha \theta^{n-1}(c_2), \dots, \alpha \theta^{n-1}(c_{n-1}), \theta^{n-1}(c_0)\right)  ,\; (d_0, d_1, \dots, d_{n-1})\rangle_E\\
        &= \alpha \langle \left(\theta^{n-1}(c_1), \theta^{n-1}(c_2), \dots,\theta^{n-1}(c_{n-1}), \alpha^{-1}\theta^{n-1}(c_0)\right)  ,\; (d_0, d_1, \dots, d_{n-1})\rangle_E\\
        &= \alpha\left(\sum_{i=1}^{n-1} \theta^{n-1}(c_i)d_{i-1}+\alpha^{-1} \theta^{n-1}(c_0)d_{n-1}\right).
    \end{align*}
    Since $\theta(0)=0,$ we have
    \begin{align*}
        \theta\left( \alpha\left(\sum_{i=1}^{n-1} \theta^{n-1}(c_i)d_{i-1}+\alpha^{-1} \theta^{n-1}(c_0)d_{n-1}\right)\right)&=0\\
        \alpha\left(\sum_{i=1}^{n-1} \theta^{n}(c_i)\theta(d_{i-1})+\alpha^{-1} \theta^{n}(c_0)\theta(d_{n-1})\right)&=0.
    \end{align*}
    Since $\alpha \in (\mathbb{F}_q^{\times})^{\theta}$ and order of $\theta$ divides $n,$ we obtain
    $\sum_{i=1}^{n-1} c_i\theta( d_{i-1})+\alpha^{-1} c_0 \theta(d_{n-1})=0,$ or equivalently, $\langle (c_0, c_1, \dots, c_{n-1}),\; (\alpha^{-1} \theta(d_{n-1}), \theta(d_0), \dots, \theta(d_{n-2}))  \rangle_E=0.$ This shows that $ (\alpha^{-1} \theta(d_{n-1}), \theta(d_0), \dots, \theta(d_{n-2})) \in \mathcal{C}^{\bot}.$
\end{proof}
The following result is due to Boucher and Ulmer \cite{BOUCHER2009}:
\begin{theorem}\cite{BOUCHER2009}
    Let $\mathcal{C}=\langle \bm{g}(x)=\sum_{i=0}^{r}g_ix^i\rangle$ be a $(\theta, 0, \alpha)$-cyclic code of length $n$ over $\mathbb{F}_q,$ where $\alpha \in (\mathbb{F}_q^{\times})^{\theta}$ and order of $\theta$ divides $n.$ Suppose $x^n-\alpha = \bm{h}(x)\bm{g}(x)\in Z(\mathbb{F}_q[x; \theta, 0]),$ where $\bm{h}(x)=\sum_{i=0}^{n-r}h_ix^i\in \mathbb{F}_q[x; \theta, 0].$ Then $\mathcal{C}^{\bot}$ is generated by 
    \begin{equation}
        \bm{h}^{\dagger}(x)=h_{n-r}+\theta(h_{n-r-1})x+\cdots+\theta^{n-r}(h_0)x^{n-r}.
    \end{equation}
\end{theorem}

\begin{theorem}\label{euclidean dual containing codn over F_q}
    Let $\mathcal{C}=\langle \bm{g}(x)\rangle$ be a $(\theta, 0, \alpha)$-cyclic code of length $n$ over $\mathbb{F}_q,$ where $\alpha \in (\mathbb{F}_q^{\times})^{\theta}$ and the order of $\theta$ divides $n.$ Suppose $x^n-\alpha = \bm{h}(x)\bm{g}(x)$ where $\bm{h}(x)\in \mathbb{F}_q[x; \theta, 0].$ Then $\mathcal{C}^{\bot}\subseteq \mathcal{C}$ if and only if $\bm{g}(x)$ is a right divisor of $\bm{h}^{\dagger}(x),$ where $\bm{h}^{\dagger}(x)$ is the generator polynomial of $\mathcal{C}^{\bot}.$
\end{theorem}
\begin{proof}
Let us assume that $\mathcal{C}^{\bot}\subseteq \mathcal{C}.$ Then $\bm{h}^{\dagger}(x) \in \mathcal{C}.$ Thus, $\bm{h}^{\dagger}(x)=\bm{r}(x)\bm{g}(x)$ for some $\bm{r}(x)\in \mathbb{F}_q[x; \theta, 0]. $ Hence, $\bm{g}(x)$ is a right divisor of $\bm{h}^{\dagger}(x).$
\par
Conversely, let us assume that  $\bm{g}(x)$ is a right divisor of $\bm{h}^{\dagger}(x).$ Then, $\bm{h}^{\dagger}(x)=\bm{m}(x)\bm{g}(x)$ for some $\bm{m}(x)\in \mathbb{F}_q[x; \theta, 0].$ Now, for any $\bm{a}(x)\in \mathcal{C}^\bot=\langle \bm{h}^\dagger(x)\rangle,$ we have, $\bm{a}(x)=\bm{s}(x)\bm{h}^\dagger(x)$ for some $\bm{s}(x)\in \mathbb{F}_q[x; \theta, 0].$ Then, $\bm{a}(x)=\bm{s}(x)\bm{m}(x)\bm{g}(x)\in \mathcal{C}.$ Hence, $\mathcal{C}^{\bot}\subseteq \mathcal{C}$.
\end{proof}

\begin{lemma}\label{Euclidean dual containment over R}
 Let $\mathcal{C}=\bm{e}_1\mathcal{C}_1\oplus\bm{e}_2\mathcal{C}_2\oplus\cdots\oplus\bm{e}_l\mathcal{C}_l$ be a $(\Theta, \bm{0}, \bm{a})$-cyclic code of length $n$ over $\mathcal{R},$ where $\Theta\in \mathcal{A},$ and $\bm{a}=(a_1, a_2, \dots, a_n).$ Then $\mathcal{C}^\bot \subseteq \mathcal{C}$ if and only if $\mathcal{C}_i^\bot\subseteq\mathcal{C}_i$ for each $1\leq i \leq l.$   
\end{lemma}
\begin{proof}
    Note that, if $\mathcal{C}=\underset{i=1}{\overset{l}{\bigoplus}}\,\bm{e}_i{\mathcal{C}_i},$ then $\mathcal{C}^\bot=\underset{i=1}{\overset{l}{\bigoplus}}\,\bm{e}_i{\mathcal{C}_i}^{\bot}.$ If each $\mathcal{C}_i^\bot \subseteq \mathcal{C}_i,$ then clearly $\mathcal{C}^\bot \subseteq \mathcal{C}.$ Conversely, assume that $\mathcal{C}^\bot \subseteq \mathcal{C}.$ Then $(\underset{i=1}{\overset{l}{\bigoplus}}\,\bm{e}_i{\mathcal{C}_i})^\bot\subseteq\underset{i=1}{\overset{l}{\bigoplus}}\,\bm{e}_i{\mathcal{C}_i}$ and consequently, $\underset{i=1}{\overset{l}{\bigoplus}}\,\bm{e}_i{\mathcal{C}_i}^{\bot} \subseteq \underset{i=1}{\overset{l}{\bigoplus}}\,\bm{e}_i\mathcal{C}_i.$ Since $\{\bm{e}_1,\bm{e}_2,\dots, \bm{e}_l\}$ is a basis of $\mathcal{R}$ over $\mathbb{F}_q,$ $\mathcal{C}_i^\bot\subseteq\mathcal{C}_i $ for each $1\leq i \leq l.$
\end{proof}
\begin{theorem}\label{dualcontaing}
     Let $\mathcal{C}=\underset{i=1}{\overset{l}{\bigoplus}}\,\bm{e}_i\mathcal{C}_i$ be a $(\Theta, \bm{0}, \bm{a})$-cyclic code of length $n$ over $\mathcal{R},$ where $\Theta=\theta_1\times\theta_2\times\cdots\times\theta_l\in \mathcal{A}$ and $\bm{a}=(a_1,a_2,\dots,a_l)\in\mathcal{R}^{\times}.$ Suppose for each $1\leq i \leq l$, $\mathcal{C}_i=\langle \bm{g}_i(x)\rangle$ and $  x^n-a_i=\bm{h}_i(x)\bm{g}_i(x)=\bm{g}_i(x)\bm{h}_i(x)$  for $\bm{h}_i(x),\bm{g}_i(x)\in \mathbb{F}_q[x;\theta_i,0].$ Then $\mathcal{C}^\bot \subseteq \mathcal{C}$ if and only if $\bm{g}_i(x)$ right divides $\bm{h}_i^\dagger(x)$ in $\mathbb{F}_q[x; \theta_i, 0]$ for each $1\leq i \leq l.$
\end{theorem}
\begin{proof}
    $\bm{h}_i^\dagger(x)$ is divisible by $\bm{g}_i(x)$ from the right in $\mathbb{F}_q[x; \theta_i, 0]$ for each $1 \leq i \leq l$ if and only if each $\mathcal{C}_i^\bot \subseteq \mathcal{C}_i$ for $1 \leq i \leq l$ (by Theorem \ref{euclidean dual containing codn over F_q}) if and only if $\mathcal{C}^\bot \subseteq \mathcal{C}$ (by Lemma \ref{Euclidean dual containment over R}).    
\end{proof}
\begin{theorem}[CSS Construction]\cite{grassl2004optimal}\label{CSS construction}
Let $\mathcal{C}_1$ and $\mathcal{C}_2$ be $[n,k_1,d_1]_q$ and $[n,k_2,d_2]_q$ over $\mathbb{F}_q$ respectively with $\mathcal{C}_2^{\bot}\subseteq \mathcal{C}_1.$ Then there exists a quantum error-correcting code $\mathcal{C}$ with parameters $[[n,k_1+k_2-n,d]]_q,$ where $d=\min\{w_H(x):x\in(\mathcal{C}_1\setminus \mathcal{C}_2^\bot)\cup(\mathcal{C}_2\setminus \mathcal{C}_1^\bot)\}.$
Furthermore, if $\mathcal{C}_1^\bot \subseteq \mathcal{C}_1,$ then there exists a quantum code with parameters $[[n,2k_1-n,d_1]]_q,$ where $d_1=\min\{w_H(x):x \in \mathcal{C}_1\setminus \mathcal{C}_1^{\bot}\}.$ 
\end{theorem}
\begin{theorem}\label{Quantum code construction}
Let $\mathcal{C}=\underset{i=1}{\overset{l}{\bigoplus}}\,\bm{e}_i\mathcal{C}_i$ be a $(\Theta, \bm{0}, \bm{a})$-cyclic code of length $n$ over $\mathcal{R},$ where $\Theta=\theta_1\times\theta_2\times\cdots\times\theta_l\in \mathcal{A}$ and $\bm{a}=(a_1,a_2,\dots,a_l)\in\mathcal{R}^{\times}.$ Suppose for each $1\leq i \leq l$, $\mathcal{C}_i=\langle \bm{g}_i(x)\rangle$ and $  x^n-a_i=\bm{h}_i(x)\bm{g}_i(x)=\bm{g}_i(x)\bm{h}_i(x)$  for $\bm{h}_i(x),\bm{g}_i(x)\in \mathbb{F}_q[x;\theta_i,0].$  If $\bm{h}_i^\dagger(x)$ is right divisible by $\bm{g}_i(x)$ for each $1\leq i \leq l,$ then there exists a quantum code with parameters $[[nl,\underset{i=1}{\overset{l}{\sum}}k_i -nl,\geq d_H]]_q,$ where $d_H$ is the Hamming distance of $\Phi(\mathcal{C})$ and $k_i=n-\deg \bm{g}_i(x)$ for each $1\leq i \leq l.$ 
\end{theorem}
\begin{proof}
    Immediately follows from Theorem \ref{Phi preserves orthogonality}, Theorem \ref{dualcontaing} and Theorem \ref{CSS construction}.
\end{proof}
\subsubsection{Annihilator dual of a \texorpdfstring{$(\theta,0,\alpha)$}{}-cyclic code over \texorpdfstring{$\mathbb{F}_q$}{}}

We define a bilinear form over $\mathbb{F}_q[x;\theta,\delta]/\langle x^n-\alpha \rangle,$ where $\alpha \in (\mathbb{F}_q^{\times})^\theta$ and $o(\theta)\mid n,$ as follows:\\
For $\bm{f}(x),\bm{g}(x)\in \mathbb{F}_q[x,\theta,\delta]/\langle x^n-\alpha \rangle,$
\begin{equation}\label{equation 6.3}
 \langle \bm{f}(x)|\bm{g}(x)\rangle:=\bm{r}(0)   
\end{equation}
where $\bm{r}(x)\equiv \bm{f}(x)\bm{g}(x)\pmod{x^n-\alpha}.$
\begin{lemma}
  The bilinear form defined in Equation \eqref{equation 6.3} is non-degenerate and hence an inner product on $\mathbb{F}_q[x;\theta,\delta]/\langle x^n-\alpha \rangle.$
\end{lemma}
\begin{proof}
    Follows from Lemma 12 of \cite{bajalan2023sigma}.
\end{proof}
%\begin{remark}
    %The bilinear form defined in Equation \eqref{equation 6.3} was already considered in [] in a more general situation, namely, .... The proof of lemma is just an implication of [].
%\end{remark}
\begin{remark}
    One can define the annihilator dual $\mathcal{C}^\circ$ of a $(\theta,\delta,\alpha)$-cyclic code $\mathcal{C}$ over $\mathbb{F}_q$ in three different ways, namely
    \begin{enumerate}
        \item $\mathcal{C}^\circ_{R}=\{\bm{f}(x)\in\frac {\mathbb{F}_q[x;\theta,\delta]}{\langle x^n-\alpha \rangle}| \langle \bm{f}(x)|\bm{c}(x)\rangle=0, \,\, \forall \bm{c}(x) \in \mathcal{C} \}$
        \item $\mathcal{C}^\circ_{L}=\{\bm{f}(x)\in\frac {\mathbb{F}_q[x;\theta,\delta]}{\langle x^n-\alpha \rangle}| \langle \bm{c}(x)|\bm{f}(x)\rangle=0, \,\, \forall \bm{c}(x) \in \mathcal{C} \}$
        \item $\mathcal{C}^\circ=\{\bm{f}(x)\in\frac {\mathbb{F}_q[x;\theta,\delta]}{\langle x^n-\alpha \rangle}| \langle \bm{c}(x)|\bm{f}(x)\rangle=\langle \bm{f}(x)|\bm{c}(x)\rangle=0, \,\, \forall \bm{c}(x) \in \mathcal{C} \}$
    \end{enumerate}
  We see from the following example that the spaces in the above remark are a valid candidate for the annihilator dual of $\mathcal{C}$ as none of them are left $\mathbb{F}_q[x; \theta,\delta]$-submodule of ${\mathbb{F}_q[x;\theta,\delta]}/{\langle x^n-\alpha \rangle}.$
  \begin{example}
       Let $\mathbb{F}_4:=\mathbb{F}_2[w],$ where $w^2+w+1=0.$ Consider the quotient ring $\frac {\mathbb{F}_4[x;\sigma_2,0]}{\langle x^6-1 \rangle},$ where $\sigma_2$ is the Frobenius automorphism of $\mathbb{F}_4.$ Suppose that $\mathcal{C}=\langle x^2+w \rangle.$ It is easy to see that none of $\mathcal{C}^\circ_{R}, \mathcal{C}^\circ_L$ and $\mathcal{C}^\circ$ is closed under left multiplication by $x.$
  \end{example}
  However, for $\theta =\Id_{\mathbb{F}_q}$ (consequently, $\delta=0$), the annihilator dual $\mathcal{C}^\circ$ carries the same structure as $\mathcal{C}$ over $\mathbb{F}_q$ (see \cite{FotueTabue}, Proposition 3) and a construction of its generator polynomial is known (see \cite{Weiqi}, Lemma 2.4). A necessary and sufficient condition for a $(\Id_{\mathbb{F}_q},0,\alpha)$-cyclic code over $\mathbb{F}_q$ to be annihilator dual-containing can be derived from Theorem 12 of \cite{bajalan2023sigma}.
\end{remark} 
\begin{theorem}\label{Ann dual containing over F_q}
    Let $\mathcal{C}=\langle \bm{g}(x) \rangle $ be a $(\Id_{\mathbb{F}_q},0,\alpha)$-cyclic code of length $n$ over $\mathbb{F}_q,$ where $\alpha\in \mathbb{F}_q^{\times}.$ Suppose $x^n-\alpha=\bm{h}(x)\bm{g}(x).$ Then $\mathcal{C}^\circ \subseteq \mathcal{C}$ if and only if $\bm{g}(x)$ is a divisor of $\bm{h}(x).$
\end{theorem}
The following theorem is a direct consequence of Euclidean CSS construction (Theorem \ref{CSS construction}) for annihilator inner product, which is referred to as annihilator CSS construction in \cite{bajalan2024polycyclic}.
\begin{theorem}[Annihilator CSS construction]\label{annCSSconstruction}
Let $\mathcal{C}_1$ and $\mathcal{C}_2$ be $[n,k_1,d_1]_q$ and $[n,k_2,d_2]_q$ over $\mathbb{F}_q$ respectively with $\mathcal{C}_2^{\circ}\subseteq \mathcal{C}_1.$ Then there exists a quantum error-correcting code $\mathcal{C}$ with parameters $[[n,k_1+k_2-n,d]]_q,$ where $d=\min\{w_H(c): c\in(\mathcal{C}_1\setminus \mathcal{C}_2^\circ)\cup(\mathcal{C}_2A\setminus \mathcal{C}_1^\circ A)\},$ where $A$ is the Gram matrix associated to the annihilator inner product. 
\end{theorem}
\begin{theorem}\label{Ann dual containing condition}
Let $\mathcal{C}=\langle \bm{g}(x)\rangle$ be a $(\Id_{\mathbb{F}_q}, 0, \alpha)$-cyclic code having parameters $[n,k,d]_q$ over $\mathbb{F}_q,$ and $x^n-\alpha=\bm{h}(x)\bm{g}(x).$ If $\bm{h}(x)$ is divisible by $\bm{g}(x),$ then there exists a quantum code with parameters $[[n,2k -n,\geq d]]_q.$ 
\end{theorem}
\begin{proof}
    Note that the Gram matrix associated to the annihilator inner product is $$A=\begin{pmatrix}
        1& 0& \dots & 0&0\\
        0& 0& \dots &0 & \alpha\\
        0& 0& \dots &\alpha &0\\
        \vdots& \vdots& \vdots& \vdots &\vdots\\
        0& \alpha &\dots& 0 &0
    \end{pmatrix}.$$ Hence, from Theorem \ref{Ann dual containing over F_q} and Theorem \ref{annCSSconstruction}, there exists a quantum code with parameters $[[n,2k -n,d^{'}]]_q,$ where $d^{'}=\min\{wt(c): c\in (\mathcal{C}\setminus\mathcal{C}^{\circ})\cup (\mathcal{C}A\setminus\mathcal{C}^{\circ}A)\}.$ Observe that $wt(\mathcal{C})=wt(\mathcal{C}A)$ and $wt(\mathcal{C}^{\circ})=wt(\mathcal{C}^{\circ}A).$ Hence, $d^{'}=\min\{wt(c): c\in \mathcal{C}\setminus \mathcal{C}^{\circ}\}\ge d.$  
\end{proof}
\section{Examples and Tables}\label{Section 7}
In this section, we present three tables. In Table \ref{zero derivation table}, we present certain MDS and almost MDS $q$-ary quantum codes obtained by Theorem \ref{Quantum code construction}. Table \ref{nonzeroderivationtable} shows certain optimal and MDS linear codes over $\mathbb{F}_q$ as Gray images of $(\Theta, \Delta_{\Theta, \bm{s}}, \bm{a})$-cyclic codes over $\mathcal{R}$ with non-zero derivations. Table \ref{Ann Dual table} contains several MDS and almost MDS $q$-ary quantum codes obtained by Theorem \ref{Ann dual containing condition}. All the computations were either done in MAGMA \cite{magma} or in SageMath \cite{sage}, and the optimality of codes was verified using https://codeTables.de \cite{Grassl:codetables}.
\begin{example}
    Consider the field $\mathbb{F}_4:=\mathbb{F}_2[w],$ where $w^2+w+1=0.$ For $q=4$ and $l=2,$ consider the ring $\mathcal{R}_2:=\mathbb{F}_4^2.$ Let $\Theta:=\sigma_2\times\sigma_2\in \Aut(\mathcal{R}_2),$ where $\sigma_2$ is the Frobenius automorphism of $\mathbb{F}_4,$ and let $\Delta_{\Theta,(w,w)}:=(\delta_{\sigma_2, w}, \delta_{\sigma_2, w})$ be a $\Theta$-Derivation of $\mathcal{R}_2$ and $\bm{a}=(1,1)\in \mathcal{R}_2^\times.$ Suppose $\mathcal{C}=\bm{e}_1\mathcal{C}_1\oplus\bm{e}_2\mathcal{C}_2$ be a $(\Theta, \Delta_{\Theta, (w,w)}, \bm{a})$-cyclic code of length $6$ over $\mathcal{R}_2.$ Then by Theorem \ref{main theorem}, for $i=1,2,\,\, \mathcal{C}_i$ is a $(\sigma_2, \delta_{\sigma_2, w} , 1)$-cyclic code of length $6$ over $\mathbb{F}_4.$ Now,
    \begin{align*}
        x^6-1&=\left(x +1\right)\times\left(x^5+x^4+x^3+x^2+x+1\right)\\
        &=\left((w + 1)x^2 + (w + 1)x + w + 1\right)\times
        \left(wx^4+wx^3+wx+w\right)
    \end{align*}
    in $\mathbb{F}_4[x; \sigma_2, \delta_{\sigma_2, w}].$ If $\mathcal{C}_1:=\langle\bm{g}_1(x)=x^5+x^4+x^3+x^2+x+1\rangle, \mathcal{C}_2:=\langle\bm{g}_2(x)=wx^4+wx^3+wx+w\rangle$ and $M_1=M_2=M=\begin{pmatrix}
    1&1\\
    1&w+1 \\
    \end{pmatrix},$ then $\Phi(\mathcal{C})$ is an optimal $[12, 3, 8]_4$ linear code over $\mathbb{F}_4.$ All the above computations were done using SageMath \cite{sage}.
\end{example}
\begin{example}
    Consider the field $\mathbb{F}_8:=\mathbb{F}_2[w],$ where $w^3+w+1=0.$ For $q=8$ and $l=2,$ consider the ring $\mathcal{R}_2:=\mathbb{F}_8^2.$ Let $\sigma_2\times\sigma_2^2\in \Aut(\mathcal{R}_2),$ where $\sigma_2$ is the Frobenius automorphism of $\mathbb{F}_8$ and $\bm{a}=(1,1)\in \mathcal{R}_2^\times.$ Suppose $\mathcal{C}=\bm{e}_1\mathcal{C}_1\oplus\bm{e}_2\mathcal{C}_2$ be a $(\sigma_2\times\sigma_2^2, 0, \bm{a})$-cyclic code of length $6$ over $\mathcal{R}_2.$ Then by Theorem \ref{main theorem}, $\mathcal{C}_1$ is a $(\sigma_2, 0, 1)$-cyclic code of length $6$ over $\mathbb{F}_8$ and $\mathcal{C}_2$ is a $(\sigma_2^2, 0, 1)$-cyclic code of length $6$ over $\mathbb{F}_8.$ Now,\\
    $$x^6-1=\left((w^2 + 1)x^3 + (w + 1)x^2 + (w^2 + 1)x + w + 1\right)\times\left(wx^3+wx^2+(w^2+w)x+w^2+w\right)$$
    in $\mathbb{F}_8[x; \sigma_2, 0],$ and
    \begin{multline*}
        x^6-1=\left((w + 1)x^3 + (w^2 + 1)x^2 + (w^2 + w)x + w\right)\times\\
        \left((w^2 + w)x^3 + (w^2 + w)x^2 + (w^2 + 1)x + w^2 + 1\right) \,\,\text{in}\,\, \mathbb{F}_8[x; \sigma_2^2, 0].
    \end{multline*}
    If $\mathcal{C}_1:=\langle\bm{g}_1(x)=wx^3+wx^2+(w^2+w)x+w^2+w\rangle, \mathcal{C}_2:=\langle\bm{g}_2(x)=(w^2 + w)x^3 + (w^2 + w)x^2 + (w^2 + 1)x + w^2 + 1)\rangle$ and $M_1=M_2=M=\begin{pmatrix}
    w^2+w+1&1\\
    1&w^2+w+1 \\
    \end{pmatrix},$ then $\Phi(\mathcal{C})$ is an optimal $[12, 6, 6]_8$ linear code over $\mathbb{F}_8.$ Moreover, if
    \begin{align*}
        \bm{h}_1(x)&=(w^2 + 1)x^3 + (w + 1)x^2 + (w^2 + 1)x + w + 1,\\
        \bm{h}_1^\dagger(x)&=(w + 1)x^3 + (w + 1)x^2 + (w^2 + 1)x + w^2 + 1,\\
        \bm{h}_2(x)&=(w + 1)x^3 + (w^2 + 1)x^2 + (w^2 + w)x + w,\\
        \bm{h}_2^\dagger(x)&=wx^3 + wx^2 + (w + 1)x + w + 1,
    \end{align*}
    then $x^6-1=\bm{h}_1(x)\bm{g}_1(x)=\bm{g}_1(x)\bm{h}_1(x)$ in $\mathbb{F}_8[x; \sigma_2, 0]$ and $x^6-1=\bm{h}_2(x)\bm{g}_2(x)=\bm{g}_2(x)\bm{h}_2(x)$ in $\mathbb{F}_8[x; \sigma_2^2, 0].$ Also, $\bm{g}_1(x)$ right divides $\bm{h}_1^\dagger(x)$ in $\mathbb{F}_8[x; \sigma_2, 0]$ and $\bm{g}_2(x)$ right divides $\bm{h}_2^\dagger(x)$ in $\mathbb{F}_8[x; \sigma_2^2, 0].$ Hence, by Theorem \ref{Quantum code construction}, we obtain an almost MDS quantum $[[12, 0,\ge 6]]_8$ code. All the above computations were done using SageMath \cite{sage}.
\end{example}
\begin{example}
    Consider the field $\mathbb{F}_9:=\mathbb{F}_2[w],$ where $w^2+2w+2=0.$ For $q=9$ and $l=2,$ consider the ring $\mathcal{R}_3:=\mathbb{F}_9^2.$ Let $\sigma_3\times\sigma_3\in \Aut(\mathcal{R}_3),$ where $\sigma_3$ is the Frobenius automorphism of $\mathbb{F}_9$ and $\bm{a}=(1,2)\in \mathcal{R}_3^\times.$ Suppose $\mathcal{C}=\bm{e}_1\mathcal{C}_1\oplus\bm{e}_2\mathcal{C}_2$ be a $(\sigma_3\times\sigma_3, 0, \bm{a})$-cyclic code of length $4$ over $\mathcal{R}_3.$ Then by Theorem \ref{main theorem}, $\mathcal{C}_1$ is a $(\sigma_3, 0, 1)$-cyclic code of length $4$ over $\mathbb{F}_9$ and $\mathcal{C}_2$ is a $(\sigma_3, 0, 2)$-cyclic code of length $4$ over $\mathbb{F}_9.$ Now,
    \begin{align*}
        x^4-1&=\left(2x^3 + (w + 1)x^2 + 2x + w + 1\right)\times\left(2x+w+1\right) &&\text{and}\\
        x^4-2&=\left((2w + 1)x^2 + 2wx + w + 2\right)\times        \left(2wx^2+2wx+w\right)
    \end{align*}
     in $\mathbb{F}_9[x; \sigma_3, 0].$ If $\mathcal{C}_1:=\langle\bm{g}_1(x)=2x+w+1\rangle, \mathcal{C}_2:=\langle\bm{g}_2(x)=2wx^2+2wx+w\rangle$ and $M_1=M_2=M=\begin{pmatrix}
    2w&w\\
    w&w\\
    \end{pmatrix},$ then $\Phi(\mathcal{C})$ is a $[8, 5, 4]_9$ MDS linear code over $\mathbb{F}_9.$ Moreover, if
    \begin{align*}
        \bm{h}_1(x)&=2x^3 + (w + 1)x^2 + 2x + w + 1,\\
        \bm{h}_1^\dagger(x)&=(2w + 2)x^3 + 2x^2 + (2w + 2)x + 2,\\
        \bm{h}_2(x)&=(2w + 1)x^2 + 2wx + w + 2,\\
        \bm{h}_2^\dagger(x)&=(w + 2)x^2 + (w + 2)x + 2w + 1,
    \end{align*}    
    then $x^4-1=\bm{h}_1(x)\bm{g}_1(x)=\bm{g}_1(x)\bm{h}_1(x)$ and $x^4-2=\bm{h}_2(x)\bm{g}_2(x)=\bm{g}_2(x)\bm{h}_2(x)$ in $\mathbb{F}_9[x; \sigma_3, 0].$ Also, $\bm{g}_1(x)$ right divides $\bm{h}_1^\dagger(x)$ and $\bm{g}_2(x)$ right divides $\bm{h}_2^\dagger(x)$ in $\mathbb{F}_9[x; \sigma_3, 0].$ Hence, by Theorem \ref{Quantum code construction}, we obtain a MDS quantum $[[8, 2, 4]]_9$ code. All the above computations were done using SageMath \cite{sage}.
\end{example}
\begin{example}
    Consider the field $\mathbb{F}_{17}.$  Then in $\mathbb{F}_{17}[x],$ 
    \begin{align*}
        \begin{multlined}[t]
            x^{17}-9=(x^7 + 5x^6 + x^5 + 2x^4 + 16x^3 + 2x^2 + 11x + 15)\times\\
        (x^{10} + 12x^9 + 7x^8 + 2x^7 + 11x^6 + 7x^5 + 7x^4 + 15x^3 + 11x^2 + 12x+ 13)=\bm{g}(x)\bm{h}(x).
        \end{multlined}
    \end{align*}
    Suppose $\mathcal{C}:=\langle \bm{g}(x)\rangle$ is a $(\Id_{\mathbb{F}_{17}}, 0, 9)$-cyclic code of length $17$ over $\mathbb{F}_{17}.$ Then $\mathcal{C}$ is a MDS $[17, 10, 8]_{17}$-linear code over $\mathbb{F}_{17}.$ Since $\bm{g}(x)$ divides $\bm{h}(x),$ we obtain a MDS $[[17, 3, 8]]_{17}$ quantum code by Theorem \ref{Ann dual containing condition}. 
    All the above computations were done using MAGMA \cite{magma}.
\end{example}
\begin{table}[H]
    \centering
    \begin{adjustbox}{width=1.02\textwidth}
    \begin{tabular}{|c|c|c|c|c|c|>{\centering\arraybackslash}p{7cm}|>{\centering\arraybackslash}p{2cm}|>{\centering\arraybackslash}p{2.8cm}|>
    {\centering\arraybackslash}p{2.4cm}|}
    \hline
   $q$ & $l$& $n$ & $M$& $x^n-\bm{a}$&$\Theta$ & $\bm{g}_1(x), \bm{g}_2(x), \dots, \bm{g}_l(x)$ &  Parameters of $\Phi(\mathcal{C})$ & Remarks for $\Phi(\mathcal{C})$ & Quantum Code \\ 
   \hline
   \hline
    $4$  &$2$     &$4$   &$\begin{pmatrix}
            w+1&w\\
            w&w+1 \\
        \end{pmatrix}$&$(1,1)x^4-(1,1)$ & $\sigma_2\times\sigma_2$  &$wx^2+(w+1)x+1,\,\, (w+1)x^2+wx+1$&$[8,4,4]_4$& Optimal&$[[8,0, 4]]_4^{**}$\\
   \hline
    $4$  &$2$     &$4$         &$\begin{pmatrix}
            w+1&w\\
            w&w+1 \\
        \end{pmatrix}$&$(1,1)x^4-(1,1)$&$\sigma_2\times\sigma_2$  &$wx^2+(w+1)x+1,\,\, (w+1)x+1$&$[8,5,3]_4$& Optimal&$[[8,2,\geq 3]]_4^{**}$\\
   \hline
    $4$ & $3$     &$4$         &$\begin{pmatrix}
            w^2& 1& w^2\\
            w^2&w& 0\\
            1 & w & w
        \end{pmatrix}$&$(1,1,1)x^4-(1,1, 1)$&$\sigma_2\times\sigma_2\times\sigma_2$  &$w^2x^2+x+w,\,\, wx^2+x+w^2,\,\, wx+w^2$&$[12,7,4]_4$& Optimal&$[[12,2,\geq 4]]_4$\\
   \hline
    $4$   & $2$    &8         &$\begin{pmatrix}
            w+1&w\\
            w&w+1 \\
        \end{pmatrix}$&$(1,1)x^8-(1,1)$&$\sigma_2\times\sigma_2$  &$wx^2+w,\,\, x^3+(w+1)x+w$&$[16,11,4]_4$& Optimal&$[[16,6,\geq 4]]_4$\\
   \hline
    $4$ &$2$      &$10$           &$\begin{pmatrix}
            w+1&w\\
            w&w+1 \\
        \end{pmatrix}$&$(1,1)x^{10}-(1,1)$&$\sigma_2\times\sigma_2$  &$wx^4+x^3+x+w+1,\,\, wx^5+x^4+x^3+x^2+x+w$&$[20,11,6]_4$& Almost Optimal&$[[20,2,\geq 6]]_4$\\
   \hline
    $8$ & $2$       &$6$        &$\begin{pmatrix}
            w^2+w+1&1\\
            1&w^2+w+1 \\
        \end{pmatrix}$&$(1,1)x^6-(1,1)$&$\sigma_2^2\times\sigma_2^2$  &$(w+1)x^2+w^2x+1,\,\, w^2x^2+(w^2+w+1)x+1$&$[12,8,4]_8$&  Optimal&$[[12,4,\geq 4]]_8$\\
   \hline
    $8$  & $2$      &$6$        &$\begin{pmatrix}
             w^2+w+1&1\\
            1&w^2+w+1 \\
        \end{pmatrix}$&$(1,1)x^{6}-(1,1)$&$\sigma_2\times\sigma_2^2$  &$wx^3+wx^2+(w^2+w)x+w^2+w,\,\, (w^2 + w)x^3 + (w^2 + w)x^2 + (w^2 + 1)x + w^2 + 1$&$[12,6,6]_8$& Optimal&$[[12,0,\geq 6]]_8^{**}$\\
    \hline
    $9$     & $3$  &$2$        &$\begin{pmatrix}
             2w+2&2 &2\\
             1& w& w+2\\
            2&2w+1& 2w\\
        \end{pmatrix}$&$(1,1,1)x^{2}-(1,1,1)$&$\sigma_3\times\sigma_3\times \sigma_3$  &$(2w+1)x+2w,\,\, x+w+1,\,\, 2x+w+1$&$[6,3,4]_9$& MDS&$[[6,0,4]]_9^*$\\
    \hline
    $9$     & $3$  &$2$        &$\begin{pmatrix}
             2w+2&2 &2\\
             1& w& w+2\\
            2&2w+1& 2w\\
        \end{pmatrix}$&$(1,1,1)x^{2}-(1,1,1)$&$\sigma_3\times\sigma_3\times \sigma_3$  &$(2w+1)x+2w,\,\, (2w+2)x+2,\,\, 1$&$[6,4,3]_9$& MDS&$[[6,2,4]]_9^*$\\
    \hline
    $9$     & $2$  &$4$        &$\begin{pmatrix}
             2w&w\\
            w&w \\
        \end{pmatrix}$&$(1,1)x^{4}-(1,2)$&$\sigma_3\times\sigma_3$  &$2x+w+1,\,\, 2wx^2+2wx+w$&$[8,5,4]_9$& MDS&$[[8,2,4]]_9^*$\\
        \hline
    $9$   & $2$    &$6$          &$\begin{pmatrix}
             w&2w+2\\
            2w+2&2w \\
        \end{pmatrix}$&$(1,1)x^{6}-(1,1)$&$\sigma_3\times\sigma_3$  &$2wx^3+2x^2+(w+1)x+w+2,\,\, (2w+1)x^3+x^2+(w+1)x+w$&$[12,6,6]_9$& Optimal&$[[12,0,\geq 6]]_9^{**}$\\
        \hline
    $16$    & $2$   &$4$       &$\begin{pmatrix}
             w^3+w^2&w\\
            w+2&w^3+w^2 \\
        \end{pmatrix}$&$(1,1)x^{6}-(1,1)$&$\sigma_2\times\sigma_2^3$  &$(w^2+w+1)x^2+w^2+w+1,\,\, w^3x^2+(w^3+w+1)x+w+1$&$[8,4,4]_{16}$& Almost Optimal&$[[8,0,\geq 4]]_{16}^{**}$\\
        \hline
        \multicolumn{9}{l}{\normalsize Here, $w$ is the primitive element of $\mathbb{F}_q$, $\sigma_p$ is the Frobenius automorphism of $\mathbb{F}_q,$ $M\in \GL(2, \mathbb{F}_q)$ is used to define the Gray map $\Phi,$}\\
        \multicolumn{9}{l}{\normalsize $\Theta:=\theta_1\times\theta_2\times\cdots\times\theta_l\in \mathcal{A}$ and for $1\le i\le l,$ $\bm{g}_i(x)$ is a right divisor of $x^n-a_i$ in $\mathbb{F}_q[x; \theta_i, 0].$}\\
        \multicolumn{9}{l}{\normalsize $*$ represents MDS quantum code and $**$ represents almost MDS quantum code}\\
    \end{tabular}
    \end{adjustbox}
    \caption{Certain nice linear codes over $\mathbb{F}_q$ and quantum codes obtained from Gray images of $(\Theta, \bm{0}, \bm{a})$-cyclic codes over $\mathcal{R}:=\mathbb{F}_q^l$ using Theorem \ref{Quantum code construction} }
    \label{zero derivation table}
\end{table}
\begin{table}[H]
    \centering
    \begin{adjustbox}{width=1.01\textwidth}
     \begin{tabular}{|c|c|c|c|c|c|c|c|c|c|}
    \hline
   $q$ & $l$& $n$ & $M$& $x^n-{\bm{a}}$&$\Theta$ & $\bm{g}_1(x), \bm{g}_2(x), \dots, \bm{g}_l(x)$ &  Parameters of $\Phi(\mathcal{C})$ & Remarks for $\Phi(\mathcal{C})$\\ 
   \hline
   \hline
    $4$  &$2$     &$6$   &$\begin{pmatrix}
            1&1\\
            1&w+1 \\
        \end{pmatrix}$&$(1,1)x^4-(1,1)$ & $\sigma_2\times\sigma_2$  &$wx^2+wx+w,\,\, x^2+1$&$[12,8,4]_4$& Optimal\\
   \hline
    $4$  &$2$     &$6$   &$\begin{pmatrix}
            1&1\\
            1&w+1 \\
        \end{pmatrix}$&$(1,1)x^4-(1,1)$ & $\sigma_2\times\sigma_2$  &$x^5+x^4+x^3+x^2+x+1,wx^4+wx^3+wx+w$&$[12,3,8]_4$& Optimal\\
   \hline
   $4$  &$2$     &$6$   &$\begin{pmatrix}
            1&1\\
            1&w \\
        \end{pmatrix}$&$(1,1)x^4-(1,1)$ & $\sigma_2\times\sigma_2$  &$x+w,wx^4+wx^3+wx+w$&$[12,7,4]_4$& Optimal\\
   \hline
   $8$  &$2$     &$3$   &$\begin{pmatrix}
            w^2+1&w+1\\
            1&w^2+1 \\
        \end{pmatrix}$&$(1,1)x^4-(1,1)$ & $\sigma_2\times\sigma_2^2$  &$(w+1)x+w+1,x^2+x+1$&$[6,3,4]_8$& MDS\\
    \hline
   $9$  &$2$     &$4$   &$\begin{pmatrix}
            1&2w+2\\
            2w+2&1 \\
        \end{pmatrix}$&$(1,1)x^4-(1,2)$ & $\sigma_3\times\sigma_3$  &$wx^3+wx^2+wx+w,2x^3+(2w+1)x^2+(w+1)x+w$&$[8,2,7]_9$& MDS\\
    \hline
   $9$  &$2$     &$4$   &$\begin{pmatrix}
            1&2w+2\\
            2w+2&1 \\
        \end{pmatrix}$&$(1,1)x^4-(1,2)$ & $\sigma_3\times\sigma_3$  &$wx+w,2x+2w$&$[8,6,3]_9$& Optimal\\
   \hline
   $4$  &$3$     &$4$   &$\begin{pmatrix}
            0&0&1\\
            1&1&0\\
            0&1&1\\
        \end{pmatrix}$&$(1,1,1)x^4-(1,1,1)$ & $\sigma_2\times\sigma_2\times\sigma_2$  &$wx^3+wx^2+wx+w,x+1,x+1$&$[12,7,4]_4$& Optimal\\
   \hline
   $4$  &$3$     &$4$   &$\begin{pmatrix}
            0&1&1\\
            1&1&1\\
            1&0&1\\
        \end{pmatrix}$&$(1,1,1)x^4-(1,1,1)$ & $\sigma_2\times\sigma_2\times\sigma_2$  &$wx^3+wx^2+wx+w,x+1,x^3+x^2+x+1$&$[12,5,6]_4$& Optimal\\
   \hline
   $4$  &$3$     &$4$   &$\begin{pmatrix}
            0&0&1\\
            1&1&0\\
            0&1&1\\
        \end{pmatrix}$&$(1,1,1)x^4-(1,1,1)$ & $\sigma_2\times\sigma_2\times\sigma_2$  &$(w+1)x^2+w+1,(w+1)x+w+1,w$&$[12,9,3]_4$& Optimal\\
   \hline
   \multicolumn{9}{l}{\normalsize Here, $w$ is the primitive element of $\mathbb{F}_q,$ $\sigma_p$ is the Frobenius automorphism of $\mathbb{F}_q,$}\\
   \multicolumn{9}{l}{\normalsize $\Theta:=\theta_1\times\theta_2\times\cdots\times\theta_l\in \mathcal{A}$ and $\Delta_{\Theta, \bm{s}}:=(\delta_{\theta_1,w}, \delta_{\theta_2,w}\dots, \delta_{\theta_l, w})$ and  }\\
   \multicolumn{9}{l}{for $1\le i\le l,$ $\bm{g}_i(x)$ is a right divisor of $x^n-a_i$ in $\mathbb{F}_q[x; \theta_i, \delta_{\theta_i, w}].$}
    \end{tabular}
    
    \end{adjustbox}
    \caption{Certain nice linear codes over $\mathbb{F}_q$ obtained from Gray images of $(\Theta, \Delta_{\Theta,\bm{s}}, \bm{a})$-cyclic codes over $\mathcal{R}:=\mathbb{F}_q^l$ using Theorem \ref{Graymap} }
    \label{nonzeroderivationtable}
\end{table}

\begin{table}[H]
    \centering
    \begin{adjustbox}{width=1.02\textwidth}
    \begin{tabular}{|c|c|c|>{\raggedright\arraybackslash}p{8cm}|>
    {\centering\arraybackslash}p{3cm}|c|>{\centering\arraybackslash}p{3cm}|c|}
    \hline
   $q$ & $n$ & $x^n-\alpha$ & Generator polynomial $\bm{g}(x)$ & Parameters of $\mathcal{C}:=\langle \bm{g}(x)\rangle$ & Remarks for $\mathcal{C}$ & Parameters of the Quantum Code $Q$ & Remarks for $Q$ \\
   \hline
   \hline
    $5$ & $10$ & $x^{10}-1$ & $x^3+x^2+4x+4$& $[10, 7, 3]_5$  & Optimal& $[[10,4,3]]_5$ &Almost MDS\\
   \hline
   $5$ & $20$ & $x^{20}-1$ & $x^3+3x+4$& $[20, 17, 3]_5$  & Optimal& $[[20, 14, 3]]_5$ &Almost MDS\\
   \hline
   $7$ & $7$ & $x^{7}-1$ & $x^3+4x^2+3x+6$& $[7, 4, 4]_7$  & MDS& $[[7,1, 4]]_7$ &MDS\\
   \hline
   $7$ & $7$ & $x^{7}-1$ & $x^2+5x+1$& $[7, 5, 3]_7$  & MDS& $[[7,3, 3]]_7$ &MDS\\
   \hline
   $7$ & $14$ & $x^{14}-4$ & $x^3 + 5x^2 + 3x + 1$& $[14, 11, 3]_7$  & Optimal& $[[14, 8,3]]_7$ &Almost MDS\\
   \hline
   $7$ & $14$ & $x^{14}-2$ & $x^4+6x^3+2x+3$ & $[14, 10, 4]_7$  & Optimal& $[[14, 6, 4]]_7$ &Almost MDS\\
   \hline
   $7$ & $21$ & $x^{21}-6$ & $x^3+6x^2+2x+4$ & $[21, 18, 3]_7$  & Optimal& $[[21, 15, 3]]_7$ &Almost MDS\\
   \hline
   $11$ & $11$ & $x^{11}-2$ & $x^5 + x^4 + 7x^3 + 8x^2 + 3x + 1$ & $[11, 6, 6]_{11}$  & MDS& $[[11, 1, 6]]_{11}$ & MDS\\
   \hline
   $11$ & $11$ & $x^{11}-3$ & $x^4 + 10x^3 + 10x^2 + 2x + 4$ & $[11, 7, 5]_{11}$  & MDS& $[[11, 3, 5]]_{11}$ & MDS\\
   \hline
   $11$ & $11$ & $x^{11}-10$ & $x^3+3x^2+3x+1$ & $[11, 8, 4]_{11}$  & MDS& $[[11, 5, 4]]_{11}$ & MDS\\
   \hline
   $11$ & $11$ & $x^{11}-10$ & $x^2+2x+1$ & $[11, 9, 3]_{11}$  & MDS& $[[11, 7, 3]]_{11}$ & MDS\\
   \hline
   $11$ & $22$ & $x^{22}-1$ & $x^4+9x^3+2x+10$ & $[22, 18, 4]_{11}$  &Almost MDS& $[[22, 14, 4]]_{11}$ & Almost MDS\\
   \hline
   $11$ & $22$ & $x^{22}-3$ & $x^3+5x^2+8x+7$ & $[22, 19, 3]_{11}$  &Almost MDS& $[[22, 16, 3]]_{11}$ & Almost MDS\\
   \hline
   $13$ & $13$ & $x^{13}-9$ & $x^6+11x^5+6x^4+6x^3+5x^2+8x+1$ & $[13, 7, 7]_{13}$  & MDS& $[[13, 1, 7]]_{13}$ & MDS\\
   \hline
    $13$ & $13$ & $x^{13}-8$ & $x^5+12x^4+3x^3+2x^2+5x+5$ & $[13, 8, 6]_{13}$  & MDS& $[[13, 3, 6]]_{13}$ & MDS\\
   \hline
   $13$ & $13$ & $x^{13}-9$ & $x^4+3x^3+5x^2+9x+9$ & $[13, 9, 5]_{13}$  & MDS& $[[13, 5, 5]]_{13}$ & MDS\\
   \hline
   $13$ & $13$ & $x^{13}-11$ & $x^3+6x^2+12x+8$ & $[13, 10, 4]_{13}$  & MDS& $[[13, 7, 4]]_{13}$ & MDS\\
   \hline
   $13$ & $26$ & $x^{26}-12$ & $x^4+3x^3+3x+12$ & $[26, 22, 4]_{13}$  & Almost MDS& $[[26, 18, \ge 4]]_{13}$ & Almost MDS\\
   \hline
   $17$ & $17$ & $x^{17}-9$ & $x^8 + 13x^7 + 7x^6 + 10x^5 + 15x^4 + 11x^3 + 10x^2 + x + 1$ & $[17, 9, 9]_{17}$  & MDS& $[[17, 1, 9]]_{17}$ & MDS\\
   \hline
   $17$ & $17$ & $x^{17}-9$ & $x^7 + 5x^6 + x^5 + 2x^4 + 16x^3 + 2x^2 + 11x + 15$ & $[17, 10, 8]_{17}$  & MDS& $[[17, 3, 8]]_{17}$ & MDS\\
   \hline
   $17$ & $17$ & $x^{17}-8$ & $x^6 + 3x^5 + 8x^4 + 11x^3 + 2x^2 + 14x + 4$ & $[17, 11, 7]_{17}$  & MDS& $[[17, 5, 7]]_{17}$ & MDS\\
   \hline
   $17$ & $17$ & $x^{17}-8$ & $x^5 + 11x^4 + 11x^3 + 14x^2 + 12x + 8$ & $[17, 12, 6]_{17}$  & MDS& $[[17, 7, 6]]_{17}$ & MDS\\
   \hline
   $17$ & $17$ & $x^{17}-7$ & $x^4 + 6x^3 + 5x^2 + 5x + 4$ & $[17, 13, 5]_{17}$  & MDS& $[[17, 9, 5]]_{17}$ & MDS\\
   \hline
   $17$ & $17$ & $x^{17}-6$ & $x^3 + 16x^2 + 6x + 5$ & $[17, 14, 4]_{17}$  & MDS& $[[17, 11, 4]]_{17}$ & MDS\\
   \hline
   $17$ & $17$ & $x^{17}-6$ & $x^2 + 5x + 2$ & $[17, 15, 3]_{17}$  & MDS& $[[17, 13, 3]]_{17}$ & MDS\\
   \hline
   $19$ & $19$ & $x^{19}-4$ & $x^9 + 2x^8 + 6x^7 + x^6 + 13x^5 + 5x^4 + 12x^3 + 12x^2 + 7x + 18$ & $[19, 10, 10]_{19}$  & MDS& $[[19, 1, 10]]_{19}$ & MDS\\
   \hline
   $19$ & $19$ & $x^{19}-2$ & $x^8 + 3x^7 + 17x^6 + 8x^5 + 18x^4 + 13x^3 + 6x^2 + 2x + 9$ & $[19, 11, 9]_{19}$  & MDS& $[[19, 3, 9]]_{19}$ & MDS\\
   \hline
   $19$ & $19$ & $x^{19}-4$ & $x^7 + 10x^6 + 13x^5 + 2x^4 + 11x^3 + 4x^2 + x + 13$ & $[19, 12, 8]_{19}$  & MDS& $[[19, 5, 8]]_{19}$ & MDS\\
   \hline
   $19$ & $19$ & $x^{19}-4$ & $x^6 + 14x^5 + 12x^4 + 12x^3 + 2x^2 + 12x + 11$ & $[19, 13, 7]_{19}$  & MDS& $[[19, 7, 7]]_{19}$ & MDS\\
   \hline
   $19$ & $19$ & $x^{19}-5$ & $x^5 + 13x^4 + 3x^3 + 4x^2 + 9x + 10$ & $[19, 14, 6]_{19}$  & MDS& $[[19, 9, 6]]_{19}$ & MDS\\
   \hline
   $19$ & $19$ & $x^{19}-5$ & $x^4 + 18x^3 + 17x^2 + 13x + 17$ & $[19, 15, 5]_{19}$  & MDS& $[[19, 11, 5]]_{19}$ & MDS\\
   \hline
   $19$ & $19$ & $x^{19}-6$ & $x^3 + x^2 + 13x + 12$ & $[19, 16, 4]_{19}$  & MDS& $[[19, 13, 4]]_{19}$ & MDS\\
   \hline
   $19$ & $19$ & $x^{19}-6$ & $x^2 + 7x + 17$ & $[19, 17, 3]_{19}$  & MDS& $[[19, 15, 3]]_{19}$ & MDS\\
   \hline
   $23$ & $23$ & $x^{23}-14$ & $x^{11} + 7x^{10} + 16x^9 + 18x^8 + 2x^7 + 16x^6 + 6x^5 + 9x^4 + 6x^3 + 18x^2 + 14x + 1$ & $[23, 12, 12]_{23}$  & MDS& $[[23, 1, 12]]_{23}$ & MDS\\
   \hline
   $23$ & $23$ & $x^{23}-15$ & $x^{10} + 11x^9 + 5x^8 + 7x^7 + 6x^6 + 7x^5 + 16x^4 + 14x^3 + 19x^2 + 21x + 3$ & $[23, 13, 11]_{23}$  & MDS& $[[23, 3, 11]]_{23}$ & MDS\\
   \hline
   $23$ & $23$ & $x^{23}-16$ & $x^9 + 17x^8 + 16x^7 + 16x^6 + 7x^5 + 3x^4 + 14x^3 + 19x^2 + 16x + 15$ & $[23, 14, 10]_{23}$  & MDS& $[[23, 5, 10]]_{23}$ & MDS\\
   \hline
   $23$ & $23$ & $x^{23}-16$ & $x^8 + 10x^7 + 15x^6 + 3x^5 + 9x^4 + 9x^3 + 20x^2 + 17x + 12$ & $[23, 15, 9]_{23}$  & MDS& $[[23, 7, 9]]_{23}$ & MDS\\
   \hline
   $23$ & $23$ & $x^{23}-18$ & $x^7 + 12x^6 + 19x^5 + 5x^4 + 2x^3 + 6x^2 + 10x + 17$ & $[23, 16, 8]_{23}$  & MDS& $[[23, 9, 8]]_{23}$ & MDS\\
   \hline
   $23$ & $23$ & $x^{23}-18$ & $x^6 + 7x^5 + 7x^4 + 16x^3 + 14x^2 + 5x + 8$ & $[23, 17, 7]_{23}$  & MDS& $[[23, 11, 7]]_{23}$ & MDS\\
   \hline
   $23$ & $23$ & $x^{23}-13$ & $x^5 + 4x^4 + 11x^3 + 18x^2 + 21x + 19$ & $[23, 18, 6]_{23}$  & MDS& $[[23, 13, 6]]_{23}$ & MDS\\
   \hline
   $23$ & $23$ & $x^{23}-14$ & $x^4 + 13x^3 + 3x^2 + 18x + 6$ & $[23, 19, 5]_{23}$  & MDS& $[[23, 15, 5]]_{23}$ & MDS\\
   \hline
   $23$ & $23$ & $x^{23}-19$ & $x^3 + 12x^2 + 2x + 18$ & $[23, 20, 4]_{23}$  & MDS& $[[23, 17, 4]]_{23}$ & MDS\\
   \hline
   $23$ & $23$ & $x^{23}-19$ & $x^2 + 8x + 16$ & $[23, 21, 3]_{23}$  & MDS& $[[23, 19, 3]]_{23}$ & MDS\\
   \hline
   $29$ & $29$ & $x^{29}-19$ & $x^{14}+15x^{13}+4x^{12}+13x^{11}+15x^{10}+28x^{9}+16x^8+19x^7+ + 16x^6+28x^5+15x^4+13x^3+4x^2+15x+1 $ & $[29, 15, 15]_{29}$  & MDS& $[[29, 1, 15]]_{29}$ & MDS\\
   \hline
    \end{tabular}
    \end{adjustbox}
    \caption{MDS and almost MDS quantum codes obtained from $(\Id_{\mathbb{F}_q}, 0, \alpha)$-cyclic codes using Theorem \ref{Ann dual containing condition}}
    \label{Ann Dual table}
\end{table}
\section{Conclusion}\label{Conclusion}
    In this article, we studied $(\Theta, \Delta_{\Theta, \mathbf{s}}, \bm{a})$-cyclic codes over the product ring $\mathcal{R}:=\mathbb{F}_q^l$ with $l\ge 1$ determined by an automorphism $\Theta:=\theta_1\times\theta_2\times\cdots\times\theta_l$ and a $\Theta$-derivation $\Delta_{\Theta,\bm{s}},$ where each $\theta_j$ is an automorphism of $\mathbb{F}_q.$ We decomposed a $(\Theta, \Delta_{\Theta, \mathbf{s}}, \bm{a})$-cyclic code over $\mathcal{R}$ into skew constacyclic codes over $\mathbb{F}_q$ and obtained its generator polynomial. By defining Gray maps from $\mathcal{R}^n$ to $\mathbb{F}_q^{nl},$ we produced several optimal linear codes over $\mathbb{F}_q.$ We established a necessary and sufficient condition for a $(\Theta, \bm{0}, \bm{a})$-cyclic code over $\mathcal{R}$ to be a Euclidean dual-containing code. Lastly, as an application, we constructed MDS and almost MDS quantum codes by employing the Euclidean dual-containing and annihilator dual-containing CSS constructions. We believe that with high computational power, one may obtain many more good codes.
\section*{Declarations}
\textbf{Conflict of Interest.} All authors declare that they have no conflict of interest.

\section*{Acknowledgements}
The first author would like to acknowledge PMRF (PMRF Id: 1403187) for its financial support. The work of the second author was supported by the Council of Scientific and Industrial Research (CSIR) India, under grant no. 09/0086(13310)/2022-EMR-I. The authors thank the FIST Lab (Project No. SR/FST/MS-1/2019/45) for computational resources.

\bibliographystyle{abbrv}
\bibliography{Skew}

\begin{thebibliography}{10}

\bibitem{abualrub2012theta}
T.~Abualrub, N.~Aydin, and P.~Seneviratne.
\newblock On $\theta$-cyclic codes over $\mathbb{F}_2+ v\mathbb{F}_2$.
\newblock {\em Australas. J. Combin.}, 54:115--126, 2012.

\bibitem{bajalan2023sigma}
M.~Bajalan, I.~Landjev, E.~Mart{\'\i}nez-Moro, and S.~Szabo.
\newblock $(\sigma,\delta)$-polycyclic codes in ore extensions over rings.
\newblock {\em arXiv preprint arXiv:2312.07193}, 2023.

\bibitem{bajalan2024polycyclic}
M.~Bajalan and E.~Mart{\'\i}nez-Moro.
\newblock Polycyclic codes over serial rings and their annihilator css construction.
\newblock {\em Cryptogr. Commun.}, pages 1--24, 2024.

\bibitem{magma}
W.~Bosma, J.~Cannon, and C.~Playoust.
\newblock The {M}agma algebra system. {I}. {T}he user language.
\newblock {\em J. Symbolic Comput.}, 24(3-4):235--265, 1997.
\newblock Computational algebra and number theory (London, 1993).

\bibitem{boucher2007skew}
D.~Boucher, W.~Geiselmann, and F.~Ulmer.
\newblock Skew-cyclic codes.
\newblock {\em Appl. Algebra Engrg. Comm. Comput.}, 18:379--389, 2007.

\bibitem{BOUCHER2009}
D.~Boucher and F.~Ulmer.
\newblock Coding with skew polynomial rings.
\newblock {\em J. Symbolic Comput.}, 44(12):1644--1656, 2009.
\newblock Gröbner Bases in Cryptography, Coding Theory, and Algebraic Combinatorics.

\bibitem{boucher2014linear}
D.~Boucher and F.~Ulmer.
\newblock Linear codes using skew polynomials with automorphisms and derivations.
\newblock {\em Des. Codes Cryptogr.}, 70:405--431, 2014.

\bibitem{Boulagouaz}
M.~Boulagouaz and A.~Leroy.
\newblock {$(\sigma,\delta)$}-codes.
\newblock {\em Adv. Math. Commun.}, 7(4):463--474, 2013.

\bibitem{cao2020class}
Y.~Cao, Y.~Cao, H.~Q. Dinh, F.-W. Fu, J.~Gao, and S.~Sriboonchitta.
\newblock A class of linear codes of length 2 over finite chain rings.
\newblock {\em J. Algebra Appl.}, 19(06):2050103, 2020.

\bibitem{cao2020construction}
Y.~Cao, Y.~Cao, H.~Q. Dinh, F.-W. Fu, and F.~Ma.
\newblock Construction and enumeration for self-dual cyclic codes of even length over $\mathbb{F}_{2^{m}}+ u\mathbb{F}_{2^{m}}$.
\newblock {\em Finite Fields Appl.}, 61:101598, 2020.

\bibitem{FotueTabue}
A.~Fotue-Tabue, E.~Mart\'inez-Moro, and J.~T. Blackford.
\newblock On polycyclic codes over a finite chain ring.
\newblock {\em Adv. Math. Commun.}, 14(3):455--466, 2020.

\bibitem{gao2013skew}
J.~Gao.
\newblock Skew cyclic codes over $\mathbb{F}_p+ v\mathbb{F}_p$.
\newblock {\em J. Appl. Math. Inform.}, 31(3\_4):337--342, 2013.

\bibitem{gao2017skew}
J.~Gao, F.~Ma, and F.~Fu.
\newblock Skew constacyclic codes over the ring $\mathbb{F}_q+ v\mathbb{F}_q$.
\newblock {\em Appl. Comput. Math}, 6(3):286--295, 2017.

\bibitem{Grassl:codetables}
M.~Grassl.
\newblock {Bounds on the minimum distance of linear codes and quantum codes}.
\newblock Online available at \url{http://www.codetables.de}, 2007.
\newblock Accessed on 2024-12-10.

\bibitem{grassl2004optimal}
M.~Grassl, T.~Beth, and M.~Roetteler.
\newblock On optimal quantum codes.
\newblock {\em Int. J. Quantum Inf.}, 2(01):55--64, 2004.

\bibitem{Hammons}
A.~Hammons, P.~Kumar, A.~Calderbank, N.~Sloane, and P.~Sole.
\newblock The $\mathbb{Z}_4$-linearity of {K}erdock, {P}reparata, {G}oethals, and related codes.
\newblock {\em IEEE Trans. Inform. Theory}, 40(2):301--319, 1994.

\bibitem{jitman2010skew}
S.~Jitman, S.~Ling, and P.~Udomkavanich.
\newblock Skew constacyclic codes over finite chain rings.
\newblock {\em Adv. Math. Commun.}, 6(1):39--63, 2012.

\bibitem{jian2022}
F.~Ma, J.~Gao, and F.-W. Fu.
\newblock {$(x^n-(a+bw),\xi,\eta)$}-skew constacyclic codes over {$\mathbb{F}_q+w\mathbb{F}_q$} and their applications in quantum codes.
\newblock {\em Quantum Inf. Process.}, 21(10):Paper No. 348, 18, 2022.

\bibitem{patel2021}
S.~Patel and O.~Prakash.
\newblock {$(\theta, \delta _\theta )$}-cyclic codes over {$\mathbb{F}_q[u,v]/\langle u^2-u, v^2-v, uv-vu \rangle $}.
\newblock {\em Des. Codes Cryptogr.}, 90(11):2763--2781, 2022.

\bibitem{Weiqi}
W.~Qi.
\newblock On the polycyclic codes over {$\mathbb{F}_q+u\mathbb{F}_q $}.
\newblock {\em Adv. Math. Commun.}, 18(3):661--673, 2024.

\bibitem{Bhaintwal2018}
A.~Sharma and M.~Bhaintwal.
\newblock A class of skew-cyclic codes over {$\mathbb{Z}_4+u\mathbb{Z}_4$} with derivation.
\newblock {\em Adv. Math. Commun.}, 12(4):723--739, 2018.

\bibitem{Ashutosh2024}
A.~Singh, P.~Sharma, and O.~Prakash.
\newblock New quantum codes and $(\theta, \delta, \beta)$-cyclic codes.
\newblock {\em IEEE Access}, 12:90345--90352, 2024.

\bibitem{sage}
W.~Stein et~al.
\newblock {\em {S}age {M}athematics {S}oftware ({V}ersion 9.5)}.
\newblock The Sage Development Team, 2022.
\newblock {\tt http://www.sagemath.org}.

\end{thebibliography}

\end{document}